\newif\ifdraft \draftfalse
\newif\iffull \fulltrue
\newif\ifarxiv\arxivtrue

\makeatletter \@input{tex.flags} \makeatother

\documentclass[11pt]{article}

\usepackage{kpfonts}
		\DeclareMathAlphabet{\mathsf}{OT1}{cmss}{m}{n}
		\SetMathAlphabet{\mathsf}{bold}{OT1}{cmss}{bx}{n}
		\DeclareMathAlphabet{\mathtt}{OT1}{cmtt}{m}{n}
		\SetMathAlphabet{\mathtt}{bold}{OT1}{cmtt}{bx}{n}
		
		\renewcommand*\sfdefault{cmsf}
\usepackage{algorithm}
\usepackage[noend]{algorithmic}

\usepackage{fullpage}

\usepackage{amsmath, amssymb, amsthm}
\usepackage{mathtools}
\usepackage{verbatim}
\usepackage{color}
\usepackage{xcolor}

\usepackage{multicol}

\definecolor{DarkGreen}{rgb}{0.1,0.5,0.1}
\definecolor{DarkRed}{rgb}{0.5,0.1,0.1}
\definecolor{DarkBlue}{rgb}{0.1,0.1,0.5}
\usepackage[]{hyperref}
\hypersetup{
    unicode=false,          
    pdftoolbar=true,        
    pdfmenubar=true,        
    pdffitwindow=false,      
    pdftitle={},    
    pdfauthor={}
    pdfsubject={},   
    pdfnewwindow=true,      
    pdfkeywords={keywords}, 
    colorlinks=true,       
    linkcolor=DarkRed,          
    citecolor=DarkGreen,        
    filecolor=DarkRed,      
    urlcolor=DarkBlue,          
}
\usepackage{xspace}
\usepackage{cleveref}
\usepackage{thmtools, thm-restate}
\usepackage[sort&compress]{natbib}

\newcommand{\xhdr}[1]{\vspace{2mm} \noindent{\bf #1}}
\newcommand{\LDOTS}{\, ,\ \ldots\ ,}     
\newcommand{\term}[1]{\ensuremath{\mathtt{#1}}\xspace}
\newcommand{\BwK}{\term{BwK}} 
\newcommand{\SWtot}{\mathtt{SW}_{\mathtt{tot}}}
\newcommand{\supSWtot}{\mathtt{supSW}_{\mathtt{tot}}}
\newcommand{\ETW}{expected total welfare\xspace}
\newcommand{\bunPolicy}{bundling policy\xspace}
\newcommand{\bunPolicies}{bundling policies\xspace}
\newcommand{\Smin}{s_{\mathtt{min}}}

\newcommand{\vmax}{V_{\mathtt{max}}}
\newcommand{\dd}{\Delta_{d+1}}
\newcommand{\tv}{\tilde v}

\newcommand{\initOneLiners}{%
 	\setlength{\itemsep}{0pt}
	\setlength{\parsep }{0pt}
  	\setlength{\topsep }{0pt}     	
}

\usepackage{color-edits}
\addauthor{as}{brown}      

\newcommand{\sw}[1]{\ifdraft \textcolor{blue}{[Steven: #1]}\fi}

\newcommand\RR{\mathbb{R}}

\newcommand\cB{\mathcal{B}}

\newcommand\cF{\mathcal{F}}

\newcommand\cP{\mathcal{P}}

\newcommand\cV{\mathcal{V}}

\newcommand\cL{\mathcal{L}}
\newcommand\cD{\mathcal{D}}

\newcommand{\cI}{\mathcal{I}}

\newcommand{\SW}{\term{SW}}   

\newcommand{\pmin}{p_{\text{min}}}

\newcommand{\bl}{^\bullet}
\newcommand{\SCP}{\term{SCP}}
\newcommand{\R}{R}
\newcommand{\ROPT}{\mathrm{R\text{-}OPT}}
\newcommand{\learnprice}{\text{\sc BunToPrice}}
\newcommand{\op}{\text{\sc OWel}}
\newcommand{\opu}{\text{\sc OWel-UD}}
\DeclareMathOperator{\poly}{poly}

\DeclareMathOperator*{\Expectation}{\mathbb{E}}
\newcommand{\Ex}[2]{\Expectation_{#1}\left[#2\right]}

\newcommand{\eps}{\varepsilon}
\def\epsilon{\varepsilon}

\newcommand{\OPT}{\text{\term{VAL}}}
\renewcommand{\bar}{\overline}

\DeclareMathOperator*{\argmin}{\mathrm{argmin}}
\DeclareMathOperator*{\argmax}{\mathrm{argmax}}

\newcommand{\INDSTATE}[1][1]{\STATE\hspace{#1\algorithmicindent}}

\newtheorem{theorem}{Theorem}[section]
\newtheorem{corollary}[theorem]{Corollary}
\newtheorem{claim}[theorem]{Claim}
\newtheorem{lemma}[theorem]{Lemma}
\newtheorem{definition}[theorem]{Definition}
\newtheorem{assumption}[theorem]{Assumption}

\newtheorem{example}[theorem]{Example}

\title{Multidimensional Dynamic Pricing for Welfare Maximization}

\begin{document}

\author{ Aaron Roth\thanks{University of Pennsylvania. Email:
    \href{mailto:aaroth@cis.upenn.edu}{aaroth@cis.upenn.edu}. 
  }
  \and Aleksandrs Slivkins\thanks{Microsoft Research. Email:
    \href{mailto:slivkins@microsoft.com}{slivkins@microsoft.com}.}
  \and Jonathan Ullman\thanks{Northeastern University. Email:
    \href{mailto:jullman@ccs.neu.edu}{jullman@ccs.neu.edu}.}  \and
  Zhiwei Steven Wu\thanks{University of Pennsylvania. Email:
    \href{mailto:wuzhiwei@cis.upenn.edu}{wuzhiwei@cis.upenn.edu}. }  }


\maketitle

\begin{abstract}
  We study the problem of a seller dynamically pricing $d$ distinct
  types of indivisible goods, when faced with the online arrival of
  unit-demand buyers drawn independently from an unknown
  distribution. The goods are not in limited supply, but can only be
  produced at a limited rate and are costly to produce. The seller
  observes only the bundle of goods purchased at each day, but nothing
  else about the buyer's valuation function. Our main result is a
  dynamic pricing algorithm for optimizing welfare (including the
  seller's cost of production) that runs in time and a number of
  rounds that are polynomial in $d$ and the approximation
  parameter. We are able to do this despite the fact that (i) the
  price-response function is not continuous, and even its fractional
  relaxation is a non-concave function of the prices, and (ii) the
  welfare is not observable to the seller.

  We derive this result as an application of a general technique for
  optimizing welfare over \emph{divisible} goods, which is of
  independent interest. When buyers have strongly concave, H\"{o}lder
  continuous valuation functions over $d$ divisible goods, we give a
  general polynomial time dynamic pricing technique.  We are able to
  apply this technique to the setting of unit demand buyers despite
  the fact that in that setting the goods are not divisible, and the
  natural fractional relaxation of a unit demand valuation is not
  strongly concave. In order to apply our general technique, we
  introduce a novel price randomization procedure which has the effect
  of implicitly inducing buyers to ``regularize'' their valuations
  with a strongly concave function. Finally, we also extend our
  results to a limited-supply setting in which the number of copies of
  each good cannot be replenished.

\end{abstract}

\vfill
\thispagestyle{empty}
\setcounter{page}{0}
\pagebreak

\section{Introduction}

Consider the problem of an online retailer who sells a large variety
of  {goods}. The seller can in principle produce or procure more copies
of each good as needed, but only at a limited rate, and at some per-unit production {/procurement} cost that
varies by good. In each round, the seller can dynamically set the
price for each type of good. Each buyer has an unknown valuation function
defined  {(in general) over bundles of goods}, and quasi-linear utility for
money. Each buyer chooses which item to buy to optimize his utility
function given the prices. The seller observes the purchased  {bundle}---i.e.~the \emph{revealed preferences of the buyer}---but not the
buyer's valuation of the purchased  {bundle} (or of any other
 {bundle}). The buyer's valuation function is drawn independently from a fixed
but \emph{unknown} distribution, called the \emph{buyer
  distribution}. The seller's objective is to optimize social welfare:
the expected buyer valuation of the purchased good minus its
production cost. Social welfare, like profit, is a natural objective
for the seller: in particular, sellers attempting to grow their market
(rather than exploit an existing monopoly position) might prefer to
optimize social welfare rather than profit in the short term.

A tempting first attempt at solving this problem would be to simply
set the price for each good to be equal to its cost of production,
which would indeed maximize social welfare if there were no other
constraints on the {bundles} of items purchased by buyers. However,
this solution is unsatisfactory when additional supply can only be
generated at a bounded rate, because the cost of production bears no
relationship to the buyers' values for a good. Because of this,
setting prices equal to costs can result, for example, in every buyer
demanding {the largest possible quantity} of the same good, which the
seller may not be able to accommodate. In a more realistic setting,
there will be constraints on the rate of production and resupply for
each good. Hence, we study the welfare maximization problem in which
we impose the additional constraint that the \emph{expected bundle}
purchased (in expectation over the draw of the buyer) lies in a
bounded set. Because constraints of this sort bind \emph{across}
buyers, setting prices equal to costs fails, and the problem requires
a nontrivial solution.


Since the buyer distribution is unknown, the seller cannot directly
compute the prices that optimize social welfare. Instead, she faces a
learning problem: she can try different prices over time and observe
the responses from random buyers drawn from the distribution, and try
to learn the optimal price vector. More formally, the goal is to use a small number of rounds to learn
a price vector that nearly optimizes expected social welfare. We want the algorithm's
guarantees hold in the worst case over the choice of distributions
over buyer valuation functions.

Essentially, we are studying a welfare-optimization version of the well-known \emph{dynamic pricing} problem, also known as \emph{learn-and-earn}, with $d>1$ goods for sale. (Prior work on dynamic pricing focused on profit maximization.) At a very high level, the main challenge presented is to learn the price response function---i.e.~the function mapping prices to  {expected} bundles purchased---and then optimize it with respect to welfare. Moreover, this is a high dimensional function (for large $d$), and so one must overcome the curse of dimensionality.
Prior work on non-Bayesian dynamic pricing (e.g., \cite{BZ09,BesbesZeevi-or12,DynPricing-ec12,Broder-or12,Keskin-or14,Wang-OR14,Boer-ms14}) dealt with this challenge by making strong assumptions on the price response function itself. Typical assumptions include Lipschitzness \citep{BZ09,BesbesZeevi-or12,Wang-OR14} (which allows for discretization in low-dimensional problems), and particularly for high dimensional problems, linearity \citep{Keskin-or14,Boer-ms14} or concavity \citep{DynPricing-ec12,BZ09,Wang-OR14}.%
\footnote{Prior work that does \emph{not} make assumptions on the shape of valuations or demand curves is either restricted to selling a single good ($d=1$) \citep{Bobby-focs03,DynPricing-ec12}, or suffers from the curse of dimensionality \emph{and}  comes with performance guarantees relative to the discretized prices rather than all prices \citep{BwK-focs13}.}
However, assumptions of this sort are not well supported by a micro-economic foundation.  In fact, natural assumptions on the buyer valuations do not necessarily result in price-response function with these properties.


In this paper, we pursue a different approach which stands on stronger microeconomic foundations: we make
assumptions on the form of the valuation functions directly (and no assumptions on the distribution over valuation functions), and show that we can work with the price response function that results. This is the case \emph{despite} the fact that our problem is high dimensional, and the price response function that results from our assumptions is not concave. We also face an additional
challenge:  unlike profit, welfare is not observable, and we can observe the purchased bundle but not the buyer's valuation for that bundle. Nevertheless, we design  {algorithms} that find a
near-optimal price vector with respect to welfare, in a number of rounds that is polynomial
in $d$ and the accuracy parameter.  {(Whereas, for example, a naive solution based on discretization and Lipschitzness of the price-response function requires a number of rounds that is exponential in $d$.)} Our results also extend to the limited-supply setting.



\subsection{Our Contributions}



 {Our main result solves the problem in the setting of \emph{indivisible goods} when buyers are unit-demand or, alternatively, when the seller only allows each customer to buy a single item. Surprisingly, no other assumptions are needed! Further, we give a general result for the setting of \emph{divisible goods} under certain assumptions on the valuation functions. In fact, we show how this general result can be leveraged to yield the result for unit demands.}

 {Both settings work as follows.} There are $d$ goods. In each round, prices are set and one buyer arrives and purchases her most preferred bundle  {from the set of feasible bundles}. The seller incurs production/procurement costs for each sale, which are linear in the sold bundle.

We give a computationally efficient
and round efficient algorithm for finding a nearly
welfare-maximizing price vector subject to a constraint on the
expected consumption.  Let $\SW(p)$ be the expected social welfare
that results from setting prices $p$. The seller would like to set
prices to ensure that the expected per-round purchase of each good
$j$, denoted $x_j(p)$, is bounded above by some supply $s_j$. This
models a realistic scenario in which the seller's inventory can be
replenished, but only at a limited rate. For example, perhaps at most
one truckload of goods can be stocked per day. Approximating a
restocking period constraint with a constraint on the expected
per-customer purchase is reasonable if the restocking period
corresponds to a large number of rounds, because then the realized
consumption over these rounds concentrates around is expectation. In
the following, we will write $x(p) = (x_1(p),\dots,x_d(p))$ to denote
the bundle \emph{induced} by prices $p$.


\xhdr{Divisible goods.}  Departing from previous work, instead of
making assumptions about the functional form of the price-response
function, which depends on the buyers' valuations in aggregate, we
make assumptions on the individual buyers' valuations
themselves. Specifically, we assume the buyers' valuations are
strongly concave and H\"older continuous. (These assumptions are
satisfied by a large class of well-studied valuation,
including~\emph{CES}
and~\emph{Cobb-Douglas} as shown in~\cite{RUW16}).  The sold bundles
are constrained to lie in the bounded set $\cF\subset \RR_+^d$ (e.g.
$\cF = [0,1]^d$ means at most one unit of each good can be purchased).

\begin{theorem}[Divisible Goods]\label{thm:intro-limited-expected}
Assume divisible goods, and buyers with strongly concave and and H\"older-continuous valuations. There is an algorithm that takes as input parameters $d, \alpha, \delta > 0$ and a supply vector $s \in \RR_{>0}^d$, such that with probability at least $1-\delta$, the algorithm outputs a price vector $p \in \RR_+^d$ such that
\begin{align}\label{eq:thm:intro-limited-expected}
x(p)\leq s \quad\text{and}\quad
\SW(p) \geq \max_{p\in \RR_+^d:\;x(p)\leq s} \SW(p) - \alpha.
\end{align}
The number of rounds and the total computation time are polynomial in $d$, $\tfrac{1}{\alpha}$ and $\log\tfrac{1}{\delta}$.
\end{theorem}


\xhdr{Unit-demand buyers and indivisible goods.}
We use our result for divisible goods to give a polynomial time dynamic pricing algorithm for welfare maximization in the \emph{indivisible goods} setting when buyers have unit demand valuations.  {Here we consider \emph{distributions} $\cD$ over price vectors $p \in \RR_+^d$, rather than fixed price vectors $p$. To extend our notation, let
    $x(\cD) = \Ex{p \sim \cD}{x(p)}$
and
    $\SW(\cD) = \Ex{p \sim \cD}{\SW(p)}$.
We prove:}

\begin{theorem} [Indivisible Goods]\label{thm:intro-unit-demands}
Assume indivisible goods, and buyers with unit-demand valuations.
There is an algorithm that takes as input parameters $d, \alpha, \delta > 0$ and a supply vector $s \in \RR_+^d$, such that with probability at least $1-\delta$ the algorithm outputs a distribution $\cD$ over price vectors $p \in \RR_+^d$ such that
\begin{align}
x(\cD)\leq s \quad\text{and}\quad
\SW(\cD) \geq \max_{\text{distributions}\;\cD':\;x(\cD')\leq s} \SW(\cD') - \alpha.
\end{align}
The number of rounds and the total computation time are polynomial in $d$, $\tfrac{1}{\alpha}$ and $\log\tfrac{1}{\delta}$.
\end{theorem}

In the generality that we state our theorem, using a \emph{distribution} over prices rather than a fixed price vector is unavoidable.  {The reason has to do with how the buyers break ties when they are indifferent between goods.} As shown in \cite{tiebreaking}, without further genericity assumptions, it can be that \emph{no} fixed pricing can induce optimal (or even feasible) allocations  {if buyers use uncoordinated tie breaking rules. Instead, tie-breaking needs to be \emph{coordinated} amongst the buyers: the tie-breaking rule needs to be different for different buyers, and essentially needs to be specified by the mechanism.} The randomness in our pricing scheme serves as a coordination mechanism amongst buyers (since each buyer faces a different realization of prices).

\xhdr{Remark.}  In both settings, we prove more general theorems in
which we express our results in terms of a stronger benchmark---the
welfare of the optimal lottery over allocations, without restriction
to those that can be induced by posted pricing.%
\footnote{A posted price vector (resp., distribution over them)
  computed by an algorithm can only hope to compete with the best
  posted price vector (resp., distribution). Thus, the mathematical
  statement behind (ii) is that the posted price benchmarks used in
  the theorems are in fact equivalent to the stronger benchmark.}

The above theorems can be reformulated in terms of cumulative regret for a given time horizon $T$. Then the execution of the algorithm in the respective theorem corresponds to an \emph{exploration phase} of bounded length. The price vector $p$ computed by the algorithm is used in an \emph{exploitation phase} consisting of all subsequent rounds. \Cref{thm:intro-limited-expected} guarantees that $\op$ completes in
    $\poly(d,\log\tfrac{1}{\delta})\cdot \alpha^{-m}$
rounds, for some constant $m$. Expected regret relative to the best fixed price vector can be upper-bounded by 1 for every round of exploration, and $\alpha+\delta$ per round of exploitation. Optimizing the choice of $\alpha$ and $\delta$, we obtain regret
    $\poly(d,\log T)\cdot T^{m/(m+1)}$.
\Cref{thm:intro-unit-demands} implies a similar corollary for the unit-demands setting.


\xhdr{Extension to limited supply.}
We extend our results to a limited-supply setting.  In our model, there is a fixed horizon of $T$ rounds and the seller has a non-replenishable supply of $Ts_j$ units of each good $j$.  $T$ and $s \in [0,1]^d$ are known in advance.  Each day, the seller will set prices and a random buyer will purchase their preferred bundle until either the time horizon or the sellers' supply is exhausted, whichever comes first.

For a pricing policy $\pi$, we use $\SWtot(\pi)$ to denote its \ETW.\footnote{Considering expected welfare per round is not enough, as one pricing policy may halt sooner than another.} A ``fixed-vector"  pricing policy uses the same price vector $p$ in all rounds. Likewise, ``fixed-distribution" pricing policy always draws the price vector independently from the same fixed distribution $\cD$. The \ETW of these policies is denoted, resp., $\SWtot(p)$ and $\SWtot(\cD)$.

In the setting of divisible goods, we simply use the algorithm from \Cref{thm:intro-limited-expected} with the same constraint vector $s$. The price vector $p$ computed by this algorithm achieves high \ETW for a given problem instance: we prove that it is nearly optimal compared to the best fixed-vector pricing policy. Further, it is nearly optimal compared to \emph{any} pricing policy.

Likewise, in the setting of unit demands, we use the algorithm from \Cref{thm:intro-unit-demands} with the same constraint vector $s$. The distribution $\cD$ computed by this algorithm is nearly optimal compared to the best fixed-distribution pricing policy with $x(\cD)\leq s$.

\begin{theorem}\label{THM:INTRO-LIMITED-SUPPLY}
Consider dynamic pricing with limited supply. Fix constraint vector $s\in \RR_+^d$ and time horizon $T>32\,\log(T)/s^*$, where $s^*=\min_j s_j$.

\begin{itemize}
\item[(a)] Consider the setting of divisible goods. When the algorithm from \Cref{thm:intro-limited-expected} is given as input $d, s, \alpha, \delta$, with probability $1-\delta$ it outputs a price vector $p\in \RR_+^d$ such that
\begin{align*}
\SWtot(p) \geq \sup_{\text{pricing policies $\pi$}} \SWtot(\pi) -
    \alpha T - O\left( \sqrt{T\log(T)/s^*} \right).
\end{align*}

\item[(b)] Assume indivisible goods and unit demands.
When the algorithm from \Cref{thm:intro-unit-demands} is given as input $d, s, \alpha, \delta$, with probability $1-\delta$ it outputs distribution $\cD$ over price vectors such that
\begin{align*}
\SWtot(\cD) \geq \sup_{\text{distributions $\cD'$ with $x(\cD')\leq s$}} \SWtot(\cD') -
    \alpha T - O\left( \sqrt{T\log(T)/s^*} \right).
\end{align*}
\end{itemize}
The number of rounds and the total computation time are polynomial in $d$, $\tfrac{1}{\alpha}$ and $\log\tfrac{1}{\delta}$.
\end{theorem}

\subsection{Our Techniques}
Our general results for divisible goods build on a crucial structural property: even though the expected  welfare of the induced bundle $x(p)$ is not concave in the price vector $p$, it becomes concave if we treat the bundle itself as the decision variable. We illustrate this via a simple 1-dimensional example, adapted from~\cite{RUW16}:
\begin{example}\label{toy}
There is a single good ($d = 1$), and a single buyer with valuation $v(x) = \sqrt{x}$. Тhe seller's cost function is $c(x) = x$. If price $p$ is posted, the buyer's utility for $x$ units is $\sqrt{x} - p\cdot x$, so she would purchase
    $x^*(p) = \frac{1}{4p^2}$
units of the good. Consequently, the welfare is
    $$  \SW(p) = v(x^*(p)) - c(x^*(p)) = \frac{1}{2p} - \frac{1}{4p^2}.$$
Note that welfare is not a concave function of the price. However, if we write the welfare as a function of $x=x^*(p)$, the purchased amount of good, this function is concave:
   $$\SW(x) = v(x) - c(x) = \sqrt{x} - x.$$
\end{example}

Thus, we would like to optimize expected welfare as a function of the induced bundle. However, we only control prices and not induced bundles. To address this, our algorithm has two ``layers," where the outer layer optimizes over induced bundles, and the inner layer finds a price vector which approximately induces a given bundle. Another challenge is that welfare is not observed, since we do not observe buyer valuations. Instead, we find a way to approximate the \emph{subgradients} of welfare, and use noise-tolerant subgradient descent to optimize over the bundles.

We build on and extend the result of \cite{RUW16} for the special case of a single buyer and unlimited supply (which focuses on profit rather than welfare). The main distinction is single buyer vs. distributions over buyers; in other words, \cite{RUW16} assume that for a given price vector the outcome is deterministic, whereas in our paper it is drawn from a fixed but unknown distribution over the possible outcomes. 

The ``inner layer" of our algorithm extends the algorithm in \cite{RUW16} from  a single buyer to distributions over buyers.  This extension presents several technical challenges, and answers one of their main open questions. In particular, we analyze a generalization of the convex programming technique used in \cite{RUW16} to accommodate a distribution over (arbitrarily many) buyers. We cannot use the ``outer layer'' from \cite{RUW16} because it requires direct observations of the objective function to feed into a procedure for zeroth-order optimization, and our seller cannot directly observe the buyers' welfare (unlike profit, which is observable). Instead, we develop a new technique to obtain the subgradient for the welfare function so as to enable first-order optimization. Also, we remove a major assumption of homogenous buyer valuations.

As stated, our general result does not apply to unit demand buyers over indivisible goods. In order to cast this problem as a divisible goods problem, we view buyers as having \emph{linear} valuations over divisible goods, optimizing over the set of bundles that have at most unit $\ell_1$ norm. The bundle that maximizes a linear function is always at a vertex of the feasible region, and hence is integral.  That is, it is the bundle purchased by a unit-demand buyer in the indivisible goods setting. However, there is a substantial difficulty: our general technique relies on buyer valuation functions being \emph{strongly} concave, a condition not satisfied by linear functions.  A standard way to obtain strong convexity in the convex optimization literature is to add a strongly convex \emph{regularizer} to the objective function.  However, we do not get to modify the buyer's objective function in this way.  Instead, we perturb the price vectors proposed by our general dynamic pricing algorithm with Gumbel noise. Doing so has the property that the expected bundle purchased by each buyer (where expectation is taken over the price perturbation) is the bundle that maximizes the buyer's linear valuation function, plus an entropy regularization term~\citep{MWnote}. Thus, in expectation over our perturbations, we can view buyers as optimizing valuation functions which are strongly concave over the $\ell_1$ norm---even though for every fixed perturbation, buyers are maximizing some linear function and thus buy a unit bundle of indivisible goods. By perturbing the price vectors used over the run of our algorithm for divisible goods, therefore, we can optimize welfare over these ``regularized'' buyers. By reducing the noise rate (and hence the implicit regularization parameter), we approach the optimal welfare of the actual, unit-demand buyers.

The extension to limited supply (\Cref{THM:INTRO-LIMITED-SUPPLY}(a)) relies on a structural result about \emph{bandits with knapsacks} \citep{BwK-focs13}, a general framework of which dynamic pricing with limited supply is a special case. We use a non-standard ``embedding" of dynamic pricing into this framework, and  a concentration inequality for total welfare that requires a somewhat delicate proof.



\subsection{Related Work}

Our setting is related to several lines of work. First, dynamic pricing, a.k.a. \emph{learn-and-earn}, focuses on a seller with a large inventory of each good, facing a stream of buyers with unknown valuations. This is a large line of work, mainly in operations research --- see \cite{Boer-survey15} for a review. Most related are non-Bayesian approaches. As mentioned above, the main distinction is that we make assumptions on the buyer valuations rather than on the price response function. Also, the learn-and-earn literature does not consider welfare-optimization, to the best of our knowledge.

Second, our problem can be viewed as an instance of the \emph{multi-armed bandits} problem \citep{Gittins-book,Bubeck-survey12}, a well-studied abstract framework in which an algorithm repeatedly chooses actions (e.g., price vectors) and receives rewards (e.g., revenue from a sale). The main issue is the tension between acquisition and usage of information, a.k.a. the \emph{exploration-exploitation tradeoff}. Bandit algorithms are directly applicable to dynamic pricing either via discretization \citep{Bobby-focs03,DynPricing-ec12,BwK-focs13} or via assumptions on expected revenue.\footnote{E.g., if expected revenue is concave in prices, one can apply bandit algorithms for concave rewards \citep{FlaxmanKM-soda05,AgarwalFHKR-nips11,Hazan-nips14,Bubeck-colt15}.}
The main distinction is (again) that solutions to bandit problems tend to make assumptions directly on the rewards, in part because they do not model the finer structure behind the rewards (such as valuation functions).

Third, there are several papers on welfare-optimizing posted pricing in combinatorial auctions \citep{Balcan-ec08,Blum-focs11,Tanmoy-sicomp13,Feldman-soda15}. These papers tackle more difficult scenarios with non-divisible goods and non-IID valuations, and accordingly obtain weaker, multiplicative guarantees. Also, the pricing is either static \citep{Balcan-ec08,Feldman-soda15} (not changing over time), or changing over time but not adapting to the observed purchases
\citep{Tanmoy-sicomp13,Blum-focs11}. This research is mainly motivated by connections to mechanism design for combinatorial auctions.

Fourth, there is a large literature on \emph{revealed preferences}, starting from Samuelson \citep{sam}, see \cite{mwg,Rub12,varian} for background. Most work in economics has focused on the construction of utility functions that explain or \emph{rationalize} a given sequence of price/bundle observations, e.g. \cite{afriat}.
A recent literature studies the problem of predicting purchase decisions given past observations at different price points \citep{BV06,ZR12,BDMUV14}. More related to our paper is \cite{ACDKR15} who study the problem of iteratively setting prices to maximize the profit obtained from a single budgeted buyer (who repeatedly makes purchase decisions) with a linear utility function. The most related paper in this line is \cite{RUW16}, as discussed in the previous subsection. Ours is the first paper in this line of work able to handle indivisible goods. A related, but distinct literature focuses on \emph{learning valuation functions} from example evaluations of those functions \citep{BCIW12}, rather than from example \emph{maximizations} of those functions (as in the revealed preference literature).

\subsection{Map of the Paper}

\Cref{MAINSEC} contains the our general result for divisible goods (a
generalization of Theorem~\ref{thm:intro-limited-expected}). Our
application to unit demand valuations over indivisible goods, and the
perturbation techniques that go into deriving this result are
presented in \Cref{SEC:UNIT-DEMAND}. \ifarxiv The limited supply
setting is treated in \Cref{sec:limited-supply}.  We conclude in
\Cref{sec:conclusions}.  To improve the flow of the paper, some
details are deferred to the appendix.\fi

\section{Model and Preliminaries}
\label{sec:model}
There is a~\emph{seller} selling $d$ different types of goods to a
sequence of \emph{buyers} arriving one after another in rounds.  Each
buyer's valuation $v$ is drawn independently from an unknown
distribution $\psi$ over a finite class $\cV$ of valuation functions
over the goods, where $|\cV| =n$.%
\footnote{We take $\cV$ to be finite only for convenience. Our results
  do not depend on $n=|\cV|$, so it can be arbitrarily large.}
\footnote{Throughout, $\RR_+ = \{x \in \RR \mid x \geq 0\}$ and
  $\RR_{>0} = \{x \in \RR \mid x>0\}$ denote non-negative reals and
  positive reals, resp.  }  Both $\psi$ and $\cV$ are unknown to the
seller. Throughout, we will use $i$ to index the buyer's types in
$\cV$, and write $v_i$ for the valuation function for a buyer of type
$i$, and $\psi(v_i)$ for the probability mass on buyers of type $i$.

At each round $t$, the seller posts a price vector $p=p^t\in \RR_+^d$,
and the $t$-th buyer with valuation $v\sim \psi$ makes a purchase to
maximize his utility under these prices. In particular, we consider
two different settings: one with divisible goods and the other with
indivisible goods.

\paragraph{Divisible goods}{Each valuation
  $v\colon \RR_{+}^d \rightarrow \RR_+$ is a function from
  (fractional) bundles of goods to values. Under prices $p$, the buyer
  with valuation $v$ will purchase the utility-maximizing bundle
\[
  x^*_v(p) \equiv \argmax_{x\in\cF} [v(x) - \langle x, p \rangle]
\]
where $\cF\subset \RR_{+}^d$ denotes the set of feasible bundles
available for purchase.}

\paragraph{Indivisible goods}{We consider unit-demand buyers, who will either purchase exactly 1 unit of some good or nothing. Each
  buyer's valuation is defined by a value vector $v\in \RR_{>0}^d$
  such that $v_j$ denotes her value for 1 unit of the $j$-th good. At
  round t, a buyer with valuation $v$ purchases
\[
  x^*_v(p) \equiv \argmax_{j\in [d]\cup \{\perp\}} \left[v_{j} -
    p_j\right].
\]
where $\perp$ denotes the choice of buying nothing (we define $v_{\perp} = p_{\perp} = 0$), and we
allow arbitrary tie-breaking rules.}

The seller has a \emph{(known) cost vector} $c\in \RR_+^d$ such
that the cost of producing a unit of good $j$ is $c_j$. The seller
wishes to set prices so as to optimize the expected \emph{social
  welfare} --- the expected valuation of the buyer's purchased bundle
or item minus its production cost.  In particular, if the seller posts
a price vector $p$ over the goods, the expected social welfare is
\begin{equation} \label{swel}
\SW(p) = \Ex{v\sim \psi}{v(x_v^*(p)) - c(x^*_v(p))}.
\end{equation}
where we write $c(x^*_v(p))$ to denote the production cost for the
purchase $x^*_v(p)$. For any distribution $\cD$ over prices, the
expected welfare is defined as $\SW(\cD) = \Ex{p\sim\cD}{\SW(p)}$.

 \newcommand{\rep}{\term{ReP}}

\paragraph{Computational model}{We will think of the algorithm as having
  access to a \emph{revealed preference} oracle $\rep(\psi)$: given
  any input price vector $p\in \RR_+^d$, it will draw a random
  valuation $v$ from $\psi$, and return the purchase decision
  $x^*_v(p)$. Our goal is to design computationally efficient
  algorithms to compute optimal prices using only polynomially many
  queries to $\rep$.  Notably, the expected or realized social welfare
  is not observable to the algorithm, since it cannot observe
  $v(x_v^*(p))$.}

\subsection{Noisy Subgradient Descent}
A key ingredient in our algorithms is the ability to minimize a convex
function (or maximize a concave function), given access only to noisy
sub-gradients of the function.  We accomplish this using the gradient
descent algorithm. Below we recap some necessary background.

Let $C \subseteq \RR^d$ be a compact and convex set of diameter at
most $D$ (w.r.t. $\ell_2$ norm).  A~\emph{subgradient} of a function
$f:C\to \R$ at point $x\in C$ is any vector $g\in C$ that satisfies
the inequality $f(y) \geq f(x) + \langle g, y - x\rangle$ for any
point $y\in C$. The set of all subgradients at $x$ is denoted
$\partial f(x)$. If $f$ is differentiable, the only
subgradient $g$ is the gradient $\nabla f(x)$.

The basic subgradient descent method is an iterative algorithm that
starts at some point $x_1\in C$ and iterates the following equations
\begin{align*}
  y_{t+1} = x_t -\eta\, g_t,
\qquad \mbox{and}\qquad  x_{t+1} = \Pi_C (y_{t+1}) \label{projection}
\end{align*}
where $\eta$ is the learning rate and $g_t \in \partial f(x_t)$ is a
subgradient of the function $f$ at point $x_t$, and
$\Pi_C (x) = \argmin_{y\in C} \|x-y\|$ denote the projection operator
onto $C$.

Now, we will assume that $g_t$ and/or $x_t$ are subject to noise. We
will use two variants of the algorithm, which operate under two
different models of noise. In the first model, the algorithm only has
access to \emph{unbiased estimates} of the subgradient.

\begin{theorem}[\cite{Z03}]\label{hellyea}
  Suppose that $f$ is convex, and for some constant $D, G$, the
  estimates of the subgradients satisfy
  $\Ex{}{g_t} \in \partial f(x_t)$ and ${\|g_t\|} \leq G$ for all
  steps $t$, and the diameter of the set satisfies $\|C\| \leq
  D$. Then if we run the subgradient descent method with step size
  $\eta = D/(G\sqrt{T})$, then for any $T$ and any initial point
  $x_1\in C$, the point $z = \frac{1}{T} \sum_{t=1}^T x_t$ satisfies
\iffull  \[\fi
    \Ex{}{f(z)} \leq \min_{x\in C} f(x) + 2DG/\sqrt{T}.
\iffull  \]\fi
\end{theorem}

In the second model, the algorithm has access to
the noiseless subgradients,
but the points $x^t$ are adversarially perturbed after the projection.


\begin{restatable}{theorem}{noisyman}\label{noisyman}
  Suppose that $f$ is convex, fix constants $D, E$, and $G$. Suppose
  that the gradient descent algorithm performs the following update in
  each iteration
  \[
    y_{t+1} = x_t - \eta \, g_t \qquad \mbox{and} \qquad
    x_{t+1} = \Pi_C (y_{t+1}) + \xi_t
  \]
  such that $g_t\in \partial f(x_t)$ and $\xi_t\in \RR^d$ is a noise
  vector. Suppose that $\|g_t\| \leq G$ and $\|\xi_t\|\leq E$,
  $x_t\in C$ for all steps $t$, and the diameter of the set satisfies
  $\|C\| \leq D$. Then if we run the subgradient descent method with
  step size $\eta = D/(G\sqrt{T})$, for any $T$ and any initial point
  $x_1\in C$, the point $z = \frac{1}{T} \sum_{t=1}^T x_t$ satisfies
\iffull  \[\fi
    f(z) \leq \min_{x\in C} f(x) + DG/\sqrt{T} + GE\sqrt{T}.
\iffull  \]\fi
\end{restatable}

The proof is similar to the standard analysis of gradient descent (e.g., see Theorem 3.1 in \cite{bubeck-convexOpt-survey15}). For the sake of completeness, we provide a self-contained proof of this result in \ifarxiv\Cref{app:noisyman}\fi.


\section{A General Algorithm in the Divisible Goods Setting}
\label{MAINSEC}
This section is dedicated to the divisible good setting: we give a
computationally efficient algorithm for finding a price vector that
approximately optimizes social welfare subject to the constraint that
the expected per-round demand of each good $j$ is no more than
$s_j$. Specifically, let $x^*_\psi(p)$ denote the expected bundle
purchased by a random buyer under the prices $p$ (or the \emph{induced
  bundle} by $p$), that is
$$
 x^*_\psi(p) = \Ex{v\sim\psi}{x^*_v(p)}.
$$
Given access to the revealed preference oracle $\rep$, the algorithm
finds an approximately optimal price vector $p$ using polynomially
queries to $\rep$ and guarantees that $x^*_\psi(p)\leq s$.  Our algorithm
consists of two layers, and we present it in three main steps.

\begin{enumerate}
\item First, we analyze a pertinent convex program and derive several
  structural results. In particular, we show that the expected social
  welfare can be expressed as a concave function of the induced
  bundle.
\item Next, we present the inner layer of the algorithm: given any
  target bundle $\hat x$, we can iteratively find price vectors $p^t$
  such that the induced bundle $x_\psi^*(p^t)$ converges to $\hat x$ over
  time.
\item Finally, we show how to derive subgradients of the expected
  social welfare function from information available. The outer layer
  of the algorithm will then use (noisy) subgradient descent to
  optimize the welfare function over the bundle space.
\end{enumerate}

We make the following assumptions on the feasible set $\cF$ and each
valuation function $v\in\cV$.

\begin{assumption}[Feasible set]\label{assf}
  We assume that $\cF$ is convex, closed, has a non-empty
  interior\ifarxiv\footnote{See~\Cref{def:strictdude} for a formal
    definition.} and bounded norm: $\|\cF\|_2 \leq R$ for some
  parameter $R$\fi.%
  \footnote{For a set $C\subset \RR^d$ and a norm $\| \cdot \|$, we
    write $\| C \| = \sup_{x \in C} \|x\|$.  When the norm is
    unspecified, it is assumed to be $\ell_2$.}  A canonical example
  is $\cF = [0,1]^d$: each buyer can simultaneously buy up to one unit
  of each good.

\end{assumption}
\begin{assumption}[Valuations]\label{assv}
  Each valuation function $v$ in $\cV$ satisfies:
\begin{enumerate}
\item $v$ is monotonically increasing in each coordinate. (This can be
  relaxed to~\Cref{rm:sat}.)
\item \label{flame} $v$ is $(\lambda, \beta)$-H\"{o}lder continuous
  with respect to the $\ell_1$ norm over $\cF$, for some 
  $\lambda \geq 1$ and some absolute constant $\beta\in (0, 1]$.  Namely:
  $|v(x) - v(x') | \leq \lambda\cdot \|x - x'\|_1^\beta$ for all
  $x, x'\in \cF$.
\item \label{burning} $v$ is $\sigma$-strongly concave over $\cF$ ---
  for all $x, x'\in \cF$,
  $ v(x') \leq v(x) + \langle \nabla v(x), x' - x\rangle -
  (\sigma/2)\cdot \|x - x'\|_2^2$.
\end{enumerate}
\end{assumption}
These assumptions on the valuations are satisfied by a large class of
well-studied valuation functions, including~\emph{Constant Elasticity
  of Substitution (CES)} and~\emph{Cobb-Douglas} ({See~\cite{RUW16}
   for a proof}). We crucially rely on a use property of strongly
concave functions: any point in the domain that is close to the
minimum in objective value is also close to the minimum in Euclidean
distance\ifarxiv (see~\Cref{lem:sconvex})\fi.

\subsection{A Stochastic Convex Program}\label{scp}

Let us say that a bundle $\hat x\in \RR_+^d$ is \emph{inducible} if
there exists a price vector $p\in \RR_+^d$ such that
$x_\psi^*(p) = \hat x$. Note that each inducible bundle $\hat x$ is a
convex combination of $n$ bundles (purchased by all the buyers) in
$\cF$, so it must lie in the set $\cF$.

A centerpiece in our analysis is the following welfare maximization
convex program that characterizes the relation between the posted
prices and inducible bundles.

\begin{definition}
  For any bundle $\hat x\in \cF$, let convex program $\SCP(\hat x)$
  be the following
\begin{align}\label{ha}
  \max_{x\in \cF^n}\quad &  \sum_{v_i\in \cV} \psi(v_i)\, v_i(x_i) \, \\
\label{haha}  \mbox{such that } \quad & \sum_{v_i\in \cV} \psi(v_i)\,  x_{ij} \, \leq \hat x_j \quad\mbox{for every }j \in [d]\\
\label{hahaha}          & x_{i} \in \cF \quad \mbox{for every }v_i\in \cV
\end{align}
Let 
$\OPT(\hat x)$ be the optimal value of the convex program
$\SCP(\hat x)$. We also say that $\SCP(\hat x)$ is
\emph{supply-saturating} if its optimal solution $x\bl$ saturates all
of the supply constraints defined by~\cref{haha}, that is
$\sum_i \psi(v_i) \, x\bl_{ij} = \hat x_j$ for all $j$.
\end{definition}

To interpret the above as a stochastic welfare maximization program,
consider a market in which there are $d$ types of goods and each good
$j$ has supply $\hat x_j$. For each valuation function $v_i\in \cV$,
we introduce a buyer $i$ with this valuation, who shows up to the
market with probability $\psi(v_i)$. We use a vector
$x_i = (x_{i1}, \ldots, x_{id})\in \cF$ to represent the bundle of
goods allocated to a buyer $i$ if he shows up. Then the program is
precisely computing an allocation over all buyers to maximize the
expected welfare subject to the constraint that the expected demand is
no more than the supply given by $\hat x$.\footnote{Similar
  construction of such stochastic convex programs also appeared
  in~\cite{DSA12}.}

\begin{assumption}[Relaxing monotonicity in valuations]\label{rm:sat}
  In fact, the assumption that each valuation in the class $\cV$ is
  increasing can be relaxed. Our algorithm works as long as the class
  $\cV$ and the feasible set $\cF$ guarantees that $\SCP(\hat x)$ is
  supply-saturating for any $\hat x\in \cF$. For the sake of
  generality, our analysis will rely on the 
  supply saturation condition instead of the monotonicity of the
  valuations. This will be useful for applying the algorithm to the
  indivisible goods setting.
\end{assumption}

If the valuations in the class $\cV$ are increasing functions, the optimal solution of $\SCP(\hat x)$ will saturate all of
the supply constraints in~\cref{haha}.

\begin{claim}\label{saturation}
  Suppose that each valuation $v\in \cV$ is monotonically increasing
  in each coordinate. Then for any $\hat x\in \cF$, the convex program
  $\SCP(\hat x)$ is supply-saturating.
\end{claim}

For each of the supply constraints in~\cref{haha}, we can introduce a
dual (price) variable $p_j$ and write down the following \emph{partial
  Lagrangian}
\begin{equation}\label{lagrange}
  \cL_{\hat x}(x, p) = \sum_{v_i\in \cV} \psi(v_i) \, v_i(x_{i})  - \sum_{j=1}^d
  p_j \left( \sum_{v_i\in \cV} \psi(v_i) \, x_{ij}  - \hat x_j\right)
\end{equation}
We can also consider the \emph{Lagrange dual function} of the convex
program \iffull $g_{\hat x} \colon \RR^d \rightarrow \RR$:
\begin{equation}\label{lagrangedual}
g_{\hat x}(p) = \max_{x\in \cF^n} \cL_{\hat x}(x, p).
\end{equation}
\fi

We will mostly focus on the case where
$\hat x \in \left(\cF \cap \RR_{>0}^d\right)$, which we can show is a
sufficient condition for inducibility.\footnote{The restriction that
  the bundle be positive in each coordinate is necessary --- a bundle
  with zero in some coordinate may not be inducible. Consider the same
  simple setting in~\Cref{toy} where $d= 1$ and there is a single
  buyer with valuation $v(x) = \sqrt{x}$. Because the marginal
  valuation at 0 is infinity, there is no bounded price to induce the
  buyer to purchase 0 units of the good.}
\begin{restatable}{lemma}{timmy}\label{timmy}
  Let $\hat x\in \left( \cF \cap \RR_{>0}^d\right)$ be a
  bundle, then $\hat x$ is inducible.
\end{restatable}

\begin{proof}
  Consider the convex program $\SCP(\hat x)$. Since 
  the convex program satisfies the \emph{Slater's condition}, strong
  duality gives
\[
  \max_{x\in \cF^n} \min_{p \in \RR_+^d} \cL_{\hat x}(x, p) = \min_{p
    \in \RR_+^d}\max_{x\in \cF^n} \cL_{\hat x}(x, p) = \OPT(\hat x)
\]
Furthermore, since $\SCP(\hat x)$ is supply-saturating, the optimal
solution satisfies $\Ex{v_i\sim\psi}{x\bl_i} = \hat x$. Let $p\bl$ be
the optimal dual solution. It follows that
  \begin{align*}
    x\bl = \argmax_{x\in \cF^n} \cL(x, p\bl) &= \argmax_{x\in
      \cF^n}\sum_{i} \psi(v_i) \, v_i(x_{i}) - \sum_{j=1}^d p\bl_j
    \left( \sum_{v_i\in \cV} \psi(v_i) \, x_{ij} - \hat x_j\right)\\
    &= \argmax_{x\in \cF^n} \sum_{i} \psi(v_i) \left( v_i(x_{i}) - \langle p\bl_j,  x_{i}\rangle + \langle p\bl, \hat x\rangle\right)
  \end{align*}
  Note that the expression inside the $\argmax$ is linearly separable across
  $i$. Therefore,
  \begin{equation}\label{nips}
    x\bl_i =  \argmax_{x_i\in \cF} \left[v_i(x_i) - \langle p\bl, x_i\rangle + \langle p\bl, \hat x\rangle\right]  =\argmax_{x_i\in \cF} \left[v_i(x_i) - \langle p\bl, x_i\rangle \right]
  \end{equation}
  It follows that $x\bl_i = x^*_{v_i}(p\bl)$ for each $i$, and
  hence the price vector $p\bl$ induces the bundle $\hat x$.
\end{proof}

Next, we show that the prices that induce the bundle $\hat x$ are an
optimal solution of the Lagrangian dual, and the bundles purchased by
each buyer in response to these prices form the unique primal optimal solution.

\begin{restatable}{lemma}{ohfun}\label{ohfun}
  Let $\hat x\in \left(\cF \cap \RR_{>0}^d\right)$ be a bundle, and
  let $\hat p\in \RR_+^d$ be a price vector such that
  $x^*_\psi(\hat p) = \hat x$. Then
\begin{itemize}
\item the price vector $\hat p$ is an optimal dual solution for $\SCP(\hat x)$, and
\item the vector $x\bl\in \cF^n$ such that
  $x\bl_i = x_{v_i}^*(\hat p)$ for each $i$ is the unique optimal
  primal solution.
\end{itemize}
\end{restatable}


A very nice consequence of Lemma~\ref{ohfun} is that whenever the induced
bundle $\hat x$ is fixed, the realized bundles purchased by buyers of
each type are also fixed. This allows us to express the expected
social welfare as a function only of the induced bundle. In
particular, the expected valuation for inducing $\hat x$ in
expectation is exactly $\OPT(\hat x)$.  This suggests a different way
to express the welfare: as a function of the induced bundle (as
opposed to a function of the price vector defined in~\cref{swel}). For
each $\hat x\in \cF$, we can define
\begin{equation}\label{bundleform}
  \SW(\hat x) = \OPT(\hat x) - \langle c , \hat x\rangle.
\end{equation}

We can show that the expected social welfare for inducing $\hat x$ in
expectation is exactly $\SW(\hat x)$\ifarxiv (see~\Cref{spiderman})\fi. More
importantly, by rewriting the welfare as a function of the bundle, we
obtain a \emph{concave} objective function. This is crucial for us to
obtain an efficient algorithm later.

\begin{restatable}{lemma}{caveman}\label{caveman}
  The expected social welfare function $\SW \colon \cF \rightarrow \RR$ as defined
  in~\cref{bundleform} is concave.
\end{restatable}

With all of structural results above, we are ready to give our
two-layered algorithm for finding the welfare-maximizing prices.

\subsection{Inner Layer: Converting Target Bundles to
  Prices}\label{sec:bunprice}
Even though we can express the expected welfare as a concave function
of the induced bundle, we still cannot directly optimize the function
because the seller only controls the prices of the goods instead of
the expected induced bundle itself. To optimize over the bundle space,
we give an algorithm that finds a price vector that approximately
induces any target expected bundle $\hat x$. Specifically, suppose
that the seller has some target bundle $\hat x$ in mind, we can learn
a price vector $\hat p$ such that the expected induced bundle is close
to the target bundle:
$\left\| \hat x - {x^*_{\psi}(\hat p)} \right\| \leq \eps$.

In Lemma~\ref{ohfun}, we show that the prices that exactly induce the
target bundle $\hat x$ are the optimal dual solution for the convex
program $\SCP(\hat x)$, which is the price vector $p$ that minimizes
the Lagrangian dual function $g_{\hat x}$. We will show that if we can
find an approximate minimizer for $g_{\hat x}$, we can then
approximately induce the target expected bundle $\hat x$. In
particular, we will apply the noisy gradient descent method
(\Cref{hellyea}) to minimize the function $g_{\hat x}$, and for the
sake of convergence of the algorithm, we will restrict the search
space for the price vector to be
\begin{equation}\label{ironman}
  \cP(\eps) = \left\{p\in \RR_+^d \mid \|p\|_2 \leq \sqrt{d}
    \lambda^{(1/\beta)}
    \left(\frac{4d}{\eps^2\sigma}\right)^{(1-\beta)/\beta} \right\}
\end{equation}
where $\eps$ is the target accuracy parameter. First, we will show
that the minimax value of the Lagrangian remains close to
$\OPT(\hat x)$ even when we restrict the dual variables/prices to be
in $\cP(\eps)$.

\begin{restatable}{lemma}{sixers}\label{sixers}
\iffull
  Let $\hat x\in \left(\cF\cap \RR_{>0}^d\right)$. There exists a
  value $\ROPT$ such that
\[
  \max_{x\in \cF^n} \min_{p \in \cP(\eps)} \cL_{\hat x}(x, p) = \min_{p
    \in \cP(\eps)}\max_{x\in \cF^n} \cL_{\hat x}(x, p) = \ROPT
\]
Moreover,
$\OPT(\hat x) \leq \ROPT \leq \OPT(\hat x) + \frac{\eps^2\sigma}{4}$.
\fi
\end{restatable}

The next result translates the approximation error in minimizing the
function $g_{\hat x}$ to the error in inducing the target bundle
$\hat x$ by making use of the strong concavity of the valuations in
$\cV$.

\begin{restatable}{lemma}{westbrook}\label{westbrook}
  Let $\hat x\in \left(\cF \cap \RR_{>0}^d\right)$ and $p'$ be a price
  vector in $\cP(\eps)$ such that
  $g_{\hat x}(p') \leq \min_{p\in \cP(\eps)} g_{\hat x}(p) + \alpha$
  for some $\alpha > 0$. Let $x' = x^*_\psi(p')$ be the expected
  bundle induced by prices $p'$. Then
  $\|x' - \hat x\|_2 \leq 2\sqrt{\alpha/\sigma}.$
\end{restatable}

  Therefore, in order to (approximately) induce a target bundle in
  expectation, we just need to compute an (approximate) minimizer for
  the Lagrangian dual function $g_{\hat x}$. We first show that we can
  compute an unbiased estimate of the gradient of $g_{\hat x}$ by using
  the observed bundle purchased by a random buyer.

\begin{restatable}{lemma}{unbiasedking}\label{unbiasedking}
  Let $p\in \RR_+^d$ be any price vector, and $x_v^*(p)$ be bundle
  purchased by a buyer with valuation function $v$ under prices
  $p$. Then
  \iffull\begin{equation*}\else$\fi  \Ex{v\sim \psi}{\hat x -
      x_v^*(p)} = \hat x - x_\psi^*(p) = \nabla g_{\hat x}(p) .
    \iffull\end{equation*}\else$\fi
\end{restatable}

The result of \Cref{unbiasedking} shows that we can obtain unbiased
estimates of the gradients of the function $g_{\hat x}$ at different
prices, as long as we can obtain unbiased estimates for the expected
demand $x^*_\psi(p)$. In the next section, we will give another 
technique to obtain unbiased estimates for the gradients. Given access
to unbiased estimate of the gradients of $g_{\hat x}$, we can rely on
the noisy subgradient descent method (and its guarantee in
~\Cref{hellyea}) to minimize the function $g_{\hat x}$. Note that the
algorithm will only find a point that approximately minimizes the
function in expectation, but we can get an approximate minimizer with
high probability using a standard amplification technique --- running
the subgradient descent method for logarithmically many times, so that
one of the output price vectors is guaranteed to be accurate with high
probability. More formally:

\begin{restatable}{lemma}{harden}\label{harden}
  Let $\hat x\in \left(\cF \cap \RR_{>0}^d\right)$ be any target
  bundle. There exists an algorithm that given any target accuracy
  $\eps$ and confidence parameter $\delta$ as input, outputs a list
  $P$ of $\log(1/\delta)$ price vectors such that with probability at
  least $1-\delta$, there exists a price vector $\hat p\in P$ that
  satisfies $\left\| x_\psi^*(\hat p) - \hat x\right \| \leq
  \eps$. Furthermore, the running time, the length of the list and the
  number of queries to $\rep$ is bounded by
  $\poly(d, 1/\eps, \log(1/\delta))$.
\end{restatable}

Lastly, we have one remaining technical problem to solve: given a set
of price vectors $P$ in which at least one price vector can
approximately induce the target expected bundle $\hat x$, we need to
identify one such price vector. To accomplish this, we will simply
post each price vector $p\in P$ repeatedly, to obtain polynomially
many observations from the buyers and compute the empirical average
bundles over these polynomially many rounds. We select the price
vector whose empirical average purchased bundle is closest to the
target bundle $\hat x$. Putting all the pieces together, we obtain our
full algorithm $\learnprice$ (formal description in
\ifarxiv\Cref{alg:learnprice} in the appendix\fi).

\begin{restatable}{theorem}{bunprice}\label{bunprice}
  Let $\hat x\in \left(\cF \cap \RR_{>0}^d\right)$ be any target
  bundle. For any target accuracy parameter $\eps$ and confidence
  parameter $\delta$, the instantiation
  $\learnprice(\hat x, \eps, \delta)$ outputs a price vector $\hat p$
  that with probability at least $1 - \delta$ satisfies
  $\| \hat x - x^*_\psi(\hat p)\| \leq \eps$. Furthermore, the number
  of queries to $\rep$ is bounded by
  $\poly\left(d, 1/\eps, \log(1/\delta) \right)$.
\end{restatable}

\subsection{Outer Layer: Welfare Maximization}\label{awesomeness}
Finally, we combine the subroutine $\learnprice$ with subgradient
descent to find the welfare maximizing prices. At a high level, we
will use subgradient descent to optimize the function $\SW$ over the
bundle space, and along the way use the algorithm $\learnprice$ to
obtain prices which induce each target bundle that arises along
subgradient descent's optimization path. To ensure that the per-round
expected demand for each good $j$ is bounded by some supply $s_j$, the
algorithm will optimize over bundles in the set
$S = \{x\in \cF \mid x_j \leq s_j \mbox{ for each } j\in [d]\}$.

There are several technical challenges remaining. First, in order to
optimize the concave function $\SW$ using subgradient descent, we need
to compute a subgradient for each bundle the subgradient descent
method chooses at intermediate steps. The following result establishes
a very nice property that the price vector that induces each target
bundle $\hat x$ gives us a simple way to compute a subgradient in
$\partial \SW(\hat x)$.  In particular, this means we can obtain a
subgradient of the function $\SW$ using our subroutine $\learnprice$.

\begin{restatable}{lemma}{envelope}\label{mr.envelope}
  Let $\hat x\in \left(\RR_{>0}^d\cap \cF\right)$, and $\hat p$ be the
  price vector that induces $\hat x$. Then
  $(\hat p - c) \in \partial \SW(\hat x)$.
\end{restatable}

\paragraph{Remark.} \Cref{mr.envelope} also shows that the
welfare-optimal solution always prices every good at cost or higher,
and hence welfare-optimal solutions are always no-deficit. To see
this, note that if the prices $p$ induce expected bundle $x$, then
$(p - c)$ is a subgradient of the welfare function at $x$, where $c$
denotes the production cost vector. Hence, if for any good $j$, we had
$p_j < c_i$, the gradient would be negative in that coordinate, and we
could increase welfare by reducing the demand of good $j$,
contradicting optimality.

Second, at each iteration $t$, subgradient descent may require a
subgradient at some bundle $x^t$, but because of the error in
$\learnprice$, we only find prices to~\emph{approximately} induce the
target bundle. To overcome this issue, we will rely on the analysis of
subgradient descent under \emph{adversarial noise} (given
in~\Cref{noisyman}).

Lastly, instead of optimizing over the entire set $S$, we will
optimize over a slightly smaller set 
\[
  S_\xi = \{x\in \cF \mid \xi\leq x_j \leq s_j - \xi\}.
\]
This allows us to settle two issues: (1) we can guarantee that all of
the induced bundles lie in the set $S$ despite the error of
$\learnprice$ and (2) each bundle in $S_\xi$ is guaranteed to be
inducible since it is strictly positive in every coordinate (as
required by~\Cref{bunprice}).

\begin{restatable}{lemma}{tmac}\label{tmac}
  For any $\xi \in (0,1)$,
  $ \max_{x\in S} \SW(x) - \max_{x'\in S_\xi} \SW(x') \leq \lambda (d
  \xi)^\beta + \sqrt{d} \xi\|c\| $. 
\end{restatable}



\newcommand{\lot}{\mathrm{OPT}^\term{lot}} Putting all the pieces
together, we obtain our main algorithm $\op$ (full description presented
in \ifarxiv\Cref{alg:optprice}\fi). To establish the welfare guarantee of
$\hat \cD$, we will compare to an even stronger benchmark---the
welfare of the optimal lottery over allocations. In particular, given
any constraint vector $s\in\RR_+^{d+1}$, a feasible lottery over
allocations is a randomized mapping $\pi\colon [n] \rightarrow \cF$
that assigns each buyer to a randomized bundle such that
$\Ex{v_i\sim \psi}{\pi(i)} \leq s$.  Let $\lot$ be the optimal social
welfare achieved by a lottery over allocations. The following is the
formal guarantee of $\op$ (corresponding to~\Cref{thm:intro-limited-expected}).

\begin{restatable}{theorem}{captainamerica}\label{captain}
  For any accuracy parameter $\alpha > 0$, confidence parameter
  $\delta > 0$, and subset
  $S = \{x\in \cF \mid x_j \leq s_j \mbox{ for each } j \in [d]\}$
  given by a supply vector $s$. Given query access to $\rep$, the
  instantiation $\op(\alpha, \delta, s)$ outputs a price vector $p'$
  that with probability at least $1 - \delta$ satisfies
  $x^*_\psi(p')\leq s$ and
  \[
    \SW(p') \geq \lot - \alpha.
  \]
  Furthermore, both the run-time of the algorithm and the number of
queries to $\rep$ is bounded by
  $\poly(d,1/\alpha, 1/\sigma, 1/\lambda, \log(1/\delta))$.
\end{restatable}

\paragraph{Remark.}{ The only part of the algorithm that interacts with
  the oracle $\rep$ (or the buyers) is $\learnprice$ of the inner
  layer. Since the $\learnprice$ only requires a bounded and unbiased
  estimate of $x_\psi^*(p)$ for each price vector $p$ it queries, we
  can replace $\rep$ by any procedure that can compute such unbiased
  estimate. This is crucial for solving the unit-demand problem.}

\section{Unit-Demand Buyers with Indivisible Goods}\label{SEC:UNIT-DEMAND}
We now switch to the setting of indivisible goods and
unit-demand buyers. Our goal is to develop a computationally and query
efficient algorithm to find an approximately optimal
\emph{distribution} over prices subject to the constraint that the
per-round expected demand of each good $j$ is bounded by $s_j$.  In
particular, we will use the algorithm $\op$ in~\Cref{MAINSEC} as a
main tool for our
solution. 
Throughout, we impose the following mild boundedness assumption on the
values.
\paragraph{Assumption on valuations}{There exists a constant upper
  bound $\vmax$ such that for any buyer $i\in [n]$ and item
  $j\in [d]$, $0 < v_{ij} \leq \vmax$.}

\paragraph{Overview: relax and regularize}{A natural starting point
  for solving our problem with $\op$ is to consider the \emph{linear
    relaxation} of unit demand valuations: that is we can view the
  buyers as having linear valuation functions over divisible goods, and optimizing over a
  feasible set of bundles that is the non-negative orthant of the $\ell_1$ ball. This relaxation maintains the property that buyers buy integral quantities of each good. However, this
  approach runs into a substantial difficulty, because linear valuation functions are not \emph{strongly} concave, and strong concavity was an important ingredient in our analysis of $\op$. Instead, we will imagine that we have access to a regularized version of the
  linear relaxation of our original problem: that is, we imagine that each buyer has a regularized valuation
  function of the form $\langle v, x\rangle + \eta H(x)$, where $H$
  is the entropy function (of course, in reality, we cannot modify the valuations of the buyers). We show that we can solve the (imagined) regularized version of the problem, and also that we can induce buyers to behave (in expectation) as if their valuation functions were regularized by appropriately perturbing the price vectors we present to them. Our solution then consists of the
  following three steps:
\begin{enumerate}
\item We show that the algorithm $\op$ can compute an approximately
  optimal price vector for the regularized version of problem,
  as long as the algorithm has access to an unbiased estimate of the
  expected demand of a random regularized buyer.
\item Next, we show that we can obtain such unbiased estimates given
  access to the revealed preference oracle $\rep(\psi)$ in the
  original un-regularized instance. The key ingredient is a novel
  price perturbation technique.
\item Finally, we show how to construct an approximately optimal
  price distribution based on the price vector output by $\op$.
\end{enumerate}
}

To facilitate the discussion, we will introduce a dummy good (indexed
as $(d+1)$) to represent the buyer's option of buying nothing. The
seller's price and each buyer's value for this item is always 0, and
the per-round demand upper bound is simply $s_{d+1} = 1$. Moreover, we will
write $\Delta_{d+1}$ to denote the set of all probability distributions
over the $(d+1)$ items (or simply the simplex over the items). For any
$x\in \Delta_{d+1}$, the entropy of $x$ is defined as
$H(x) = \sum_{j\in [d+1]} x_j \log{\frac{1}{x_j}}$.  The function $H$
is strongly concave with respect to the $\ell_1$ norm over the
simplex.

\subsection{Solving the regularized problem}
Given a probability distribution $\psi$ over value vectors and a
parameter $\eta$, we imagine a corresponding $\eta$-regularized problem with
a distribution $\tilde\psi$ over valuation functions: for each
valuation vector $v_i$ in the support of $\psi$, create a regularized
valuation function $\tv_i\colon \dd \rightarrow \RR_{>0}$ with the
same probability mass $\tilde \psi_{v_i} = \psi_{v_i}$ such that
$\tv_i(x) = \langle v_i, x\rangle + \eta H(x)$.

The regularized problem is an instance of the divisible good setting,
with feasible set $\cF = \dd$. Suppose that we have access to an unbiased
estimate for $x^*_{\tilde \psi}(p)$ for any price vector $p$: then we can apply $\op$ from~\Cref{MAINSEC} to compute
approximately optimal prices. There is a
small obstacle in the analysis---the regularized valuations defined
above are not monotonically increasing in each coordinate (because of
the entropy term). However, recall that we were able to substitute monotonicity for~\Cref{rm:sat}--- that the convex program $\SCP(\hat x)$ is
supply-saturating for each $\hat x$ in our feasible region $\cF = \dd$. This is indeed satisfied in our setting:

\begin{lemma}
  Let $\cF = \dd$ and $\tilde \psi$ be a distribution over the
  regularized valuation functions of the form
  $[\langle v, x\rangle + \eta H(x)]$. Then the convex program
  $\SCP(\hat x)$ defined as
  \begin{align*}
    \max_{x\in \dd^n}\quad   \sum_{v_i\in \cV} \tilde \psi(v_i)\,\left[ \langle v_i , x_i\rangle + \eta H(x_i)\right] \, \qquad
    \mbox{ such that } \quad  \sum_{v_i\in \cV} \tilde\psi(v_i)\,  x_{ij} \, \leq \hat x_j \quad\mbox{for every }j \in [d + 1]
\end{align*}
is supply-saturating for any $\hat x\in \dd$.
\end{lemma}

\begin{proof}
  Let $x\bl$ be an optimal solution to $\SCP(\hat x)$. The convex
  combination $\Ex{\tilde\psi}{x\bl_i} \in \dd$, so we must have
  $\sum_{j\in [d+1]} \Ex{\tilde\psi}{x\bl_{ij}} = 1$. Note that $\hat x$
  also lies in $\dd$ (these are the only inducible average bundles for unit demand buyers), so $\sum_{j\in [d+1]} \hat x_j = 1$. It follows
  that all of the constraints $\Ex{\tilde\psi}{x\bl_{ij}} \leq \hat x_j$ are
  saturated.
\end{proof}

Moreover, the feasible set $\dd$ is convex, closed, has a non-empty
interior, and each of the regularized valuations is $\eta$-strongly
concave with respect to the $\ell_2$ norm,\footnote{This follows from the fact
  that the $\ell_1$ norm of any vector is bigger than its $\ell_2$
  norm. } and $((\sqrt{d+1} + \vmax, 1/2)$-H\"{o}lder continuous
(see \ifarxiv\Cref{lem:holderdude}\fi for a proof). Thus we satisfy all conditions needed to apply $\op$, so long as we have access to unbiased estimates of $x^*_{\tilde\psi}(p)$ for any price
vector $p\in \RR_{+}^{d+1}$.

\subsection{From price perturbation to value regularization}

Solving the regularized problem using $\op$ requires an unbiased
estimate for $x_{\tilde \psi}(p)$ for any price vector $p$ that the
algorithm queries, but in our problem instance the valuations are
actually drawn from $\psi$ (without the entropy term). To obtain such
an estimate, we give a price perturbation technique that allows us to
simulate the response for the regularized buyers: given any price
vector $p$, we will perturb each coordinate to
obtain a noisy price vector $p'$, with the effect that the random item purchased by
the unit-demand buyer in expectation over the perturbation equals the bundle purchased by
her regularized counterpart.\footnote{The technique of simulating
  regularization through perturbation was also used to establish the
  equivalence between ``Follow the Regularized Leader'' and ``Follow
  the Perturbed Leader''~\citep{MWnote}. In our setting, it is important that we can obtain this effect by perturbing the price vector, rather than the valuation vector, because we do not have control over buyer valuations.}

\begin{lemma}[\cite{MWnote}]\label{MWnote}
  Fix any $\eta > 0$ and any vector $u \in \RR_{>0}^{d+1}$. For each
  $j\in [d+1]$, let $G_j$ be a random number drawn independently and
  uniformly at random from $[0, 1]$, then
\[
  x^* = \argmax_{x\in \Delta_{d+1}} [\langle u, x\rangle + \eta H(x)] =
    \Ex{}{\argmax_{x\in \Delta_{d+1}} \left[\langle u, x\rangle + \eta
        \sum_{j\in [d+1]} x_j \ln \left(\ln(1/G_j)\right)\right]}.
\]
\end{lemma}
\newcommand{\conv}{\term{Convert}} The random
variable $\ln \ln(G_j)$ is distributed according to the \emph{Gumbel}
distribution. One immediate technical issue we have is that the Gumbel
distribution is unbounded, so the perturbed prices might be
negative. We also need to guarantee that the price on the dummy good
$\perp$ is
0. 
To overcome this issue, we will translate the perturbed price vector $p$ into a
non-negative price vector $p'$ with the following procedure: set
$p'_{d+1} = 0$ and
$p'_j = (p_j - p_{d+1}) - \min\{0 , \min_{j'\in [d]} (p_{j'} -
p_{d+1})\}$.  We will refer to this procedure as $\conv$, and show
that the choice made by any unit-demand buyer remains the same under
the new price vector $p'$.

\begin{restatable}{lemma}{trans}\label{trans}
  Let $p\in \RR^{d+1}$ any real-valued vector and let $p' = \conv(p)$.
  Then for any vector $v\in \RR_{>0}^{d+1}$ such that $v_{d+1} = 0$,
  \[
    \argmax_{j\in [d+1]} [v_j - p_{j}] = \argmax_{j\in [d+1]} [v_j - p'_j]
  \]
  Furthermore, $p'_j\geq 0$ for all $j\in [d]$.
\end{restatable}

Combining the Gumbel noise addition and the procedure $\conv$, we can
now obtain unbiased estimates for $x^*_{\tilde \psi}(p)$ using feedback
from $\rep(\psi)$, the revealed preference oracle for the original
problem instance. More formally, for any fixed price vector $p$,
consider the following distribution $\cD(p)$ of random prices:
\begin{enumerate}
\item For each $j\in [d+1]$, let $G_j$ be a random number drawn
  independently and uniformly at random from $[0, 1]$, and let
  $\tilde p_j = p_j + \ln \ln(G_j)$ for each $j\in [d+1]$.
\item Output $p' = \conv(\tilde p)$.
\end{enumerate}

\newcommand{\simulate}{\text{\sc Sim}}

Given this subroutine for generating random prices, we have a
procedure $\simulate$ (presented in~~\Cref{alg:siman}) for obtaining an
unbiased estimate of $x^*_{\tilde\psi}(p)$ for any price vector $p$.
We will establish the correctness of $\simulate$ in~\Cref{crucial}.

\begin{algorithm}[h]
  \caption{Simulate the response of a regularized buyer $\simulate(p, \eta)$}
 \label{alg:siman}
  \begin{algorithmic}
    \STATE{\textbf{Input:} A price vector p and regularization parameter $\eta$}

    \INDSTATE[1]{Let $p'$ a price vector drawn from $\cD(p)$}
    \INDSTATE[1]{Query the oracle $\rep(\psi)$ with $p'$ (i.e. post
      $p'$ to a random consumer) and obtain feedback $x'$}
    \STATE{\textbf{Output:} $x'$ as an unbiased estimate for
      $x^*_{\tilde\psi}(p)$}
    \end{algorithmic}
  \end{algorithm}

\begin{restatable}{lemma}{crucial}\label{crucial}
  Fix any price vector $p\in \RR_+^{d+1}$, parameter $\eta > 0$, and
  any distribution $\psi$ over value vectors in $\RR_{>0}^{d+1}$.  Let
  $x'$ be the estimate output by $\simulate(p, \eta)$, then
  $\Ex{\psi, \simulate}{x'} = x^*_{\tilde v}(p)$.
\end{restatable}

Using the subroutine $\simulate$ to obtain unbiased estimate of
$x^*_{\tilde \psi}(p)$ for different price vectors $p$, we can now
instantiate $\op$ to solve the $\eta$-regularized problem.

\subsection{Wrap-up: An approximately optimal distribution over prices}
Let $\hat p$ be the price vector output by $\op$ when solving the
$\eta$-regularized problem, and consider the distribution over prices
$\hat\cD = \cD(\hat p)$.  Similar to the divisible goods setting, we
will again compare to the optimal lottery over allocations. In
particular, the optimal lottery is given by
\begin{align*}
  \max_{x\in \dd^{n}} \sum_{i\in [n]} \psi(v_i) [ \langle v_i - c, x_i \rangle] \qquad
  \mbox{such that} \qquad \sum_{i\in [n]} \psi(v_i) x_{ij} \leq s_j \quad \mbox{ for each } j\in [d+1]
\end{align*}

\newcommand{\oeta}{\mathrm{OPT}^{\eta}}
Let $x^\star$ and $\lot$ be the optimal solution and value for the
program defined above. Since any distribution of prices is just
inducing a lottery over allocations, we know that
\[
\max_{\cD \mbox{ s.t. } x_\psi^*(\cD)\leq s} \SW(\cD) \leq \lot
\]
where we write $x_\psi^*(\cD) = \Ex{p\sim \cD, v\sim \psi}{x_v^*(p)}$
to denote expected demand over the goods.

The following lemma bounds the sub-optimality of $\hat \cD$ compared
to $\lot$ in terms of the regularization parameter $\eta$ and accuracy
guarantee of the price vector $\hat p$.

\begin{restatable}{lemma}{adderror}\label{adderror}
  Fix any regularization parameter $\eta$. Suppose that $\hat p$ is an
  $\eps$-approximately optimal price vector for the $\eta$-regularized
  problem. Then the distribution of prices $\hat \cD = \cD(\hat p)$
  satisfies
\[
  \SW(\hat \cD) \geq \lot - \eps -\eta \log(d+1).
\]
\end{restatable}

Therefore, to achieve a target accuracy of $\alpha$, it suffices to
instantiate $\op$ with accuracy parameter $\epsilon = \alpha/2$ to
solve the $\eta$-regularized problem with
$\eta = \alpha/(2 \log(d+1))$. Putting all the pieces together, we
have our algorithm $\opu$ (formally presented
in \ifarxiv\Cref{alg:optpriceUD}\fi) that achieves the following main result,
which recovers~\Cref{thm:intro-unit-demands}.

\begin{theorem}\label{thehook}
  For any accuracy parameter $\alpha > 0$, confidence parameter
  $\delta > 0$, and subset
  $S = \{x\in \cF \mid x_j \leq s_j \mbox{ for each } j \in [d]\}$
  given by a supply vector $s$. Given query access to $\rep$, the
  instantiation $\opu(\alpha, \delta, s)$ outputs a distribution
  $\hat \cD$ over prices such that with probability at least
  $1 - \delta$ satisfies $x^*_\psi(\hat \cD)\leq s$ and
  \[
    \SW(\hat \cD) \geq \lot - \alpha.
  \]
  Furthermore, both the run-time of the algorithm and the number of
  queries to $\rep$ is bounded by
  $\poly(d,1/\alpha, \vmax, \log(1/\delta))$.
\end{theorem}


\section{Limited Supply: Proof of \Cref{THM:INTRO-LIMITED-SUPPLY}}
\label{sec:limited-supply}

We turn our attention to dynamic pricing with limited supply, so as to prove \Cref{THM:INTRO-LIMITED-SUPPLY}.

\xhdr{Setting and notation.} Compared to the main model described in \Cref{sec:model}, the limited supply setting differs in the following ways. A problem instance is characterized by a pair $(s,T)$, where $s\in \RR_+^d$ is the supply vector and $T$ is the time horizon (the maximal number of rounds). Initially the seller has $T\cdot s_j$ units of each good $j$. For ease of exposition, assume that at most one unit of each good can be sold in each round. Execution halts when the time horizon is exceeded, or when the remaining supply of any one good falls below $0$. Performance of a given pricing policy $\pi$ is characterized by its \ETW (over the entire execution), denoted $\SWtot(\pi)$.

We are particularly interested in ``fixed-vector"  pricing policies: ones that always uses the same fixed price vector $p$. The \ETW of such policy is denoted $\SWtot(p)$. Likewise, ``fixed-distribution" pricing policies always draw the price vector independently from the same fixed distribution $\cD$; the \ETW of such policy is denoted $\SWtot(\cD)$.

The \emph{induced bundle} for a given price vector $p$ is a vector $x=x(p)\in \RR_+^d$ such that $x_j$ is the per-round expected consumption of each product $j$ if price $p$ is chosen. A bundle $x \in \RR_+^d$ is \emph{inducible} if $x=x(p)$ for some price vector $p$. 


\xhdr{Connection to ``Bandits with Knapsacks''.} We represent our problem as a special case of ``Bandits with Knapsacks" \citep{BwK-focs13}, a general setting for multi-armed bandit problems with resource consumption (henceforth denoted \BwK). In \BwK, there are several resources consumed by an algorithm, with a limited supply of each. In each round an algorithm chooses from a fixed set of `arms', receives a reward and consumes some resources. Thus, the outcome from choosing an arm is a vector (\emph{outcome vector}) which consists of the reward and the consumption of all resources. The outcome vector is assumed to be an independent draw of some fixed but unknown distribution that depends only on the chosen arm.

Dynamic pricing with limited supply is a canonical special case of \BwK: arms correspond to price vectors, resources correspond to the goods (one resource for each good), and rewards is the seller's utility from a given customer (typically revenue or profit, in our case --- welfare). The outcome in a given round is determined by the purchased bundle.

\subsection{Divisible goods: proof of \Cref{THM:INTRO-LIMITED-SUPPLY}(a)}

We use a different, non-standard connection to \BwK: arms correspond to inducible bundles, rather than price vectors. More precisely, for each inducible bundle $x$ 
we have an arm in \BwK such that choosing this arm means choosing a particular price vector $p_x$ that induces bundle $x$. Henceforth, such an arm is termed \emph{arm-bundle} $x$. The reward and resource consumption from choosing this arm-bundle are defined as those from choosing $p_x$. Note that the expected consumption of each good $j$ is simply $x_j$. An algorithm $\pi$ for such an instance of \BwK, i.e., an algorithm that in each round selects an inducible bundle and observes the purchased bundle, will be called a \emph{\bunPolicy}. Its \ETW is denoted $\SWtot(\pi)$. A ``fixed-bundle" \bunPolicy chooses some arm-bundle $x$ in each round. Its \ETW is denoted as $\SWtot(x)$.



Recall that choosing an arm-bundle $x$ determines the realized bundles purchased by each type of buyers. (It is a consequence of \Cref{ohfun}.) In other words, the realized bundles do not depend on the choice of $p_x$. Therefore:

\begin{claim}\label{cl:limited-supply-A-to-pi}
Fix a problem instance. For any pricing policy $\pi'$, there exists a \bunPolicy $\pi$ with the same expected total welfare.
\end{claim}

\begin{proof}
The \bunPolicy $\pi$ is constructed as follows: whenever $\pi'$ chooses a price vector $p$ that induces bundle $x$, $\pi$ chooses arm-bundle $x$. Then $\pi$ and $\pi'$ have the same distribution over the outcome vectors in each round $t$ (this is proved by induction on $t$).
\end{proof}

The analysis in \cite{BwK-focs13} emphasizes ``fixed-distribution"  \bunPolicies: where in each round the arm-bundle is sampled independently from a fixed distribution $D$ over arm-bundles. Let $\SWtot(D)$ and $\SW(D)$ denote, resp., the \ETW and the expected per-round welfare from this \bunPolicy. A structural result from \cite{BwK-focs13}, as specialized to our setting, essentially reduces optimization over arbitrary \bunPolicies to that over fixed-distribution \bunPolicies:

\begin{lemma}[specialized from \cite{BwK-focs13}]\label{lm:BwK-LP}
Fix a finite set $\cB$ of arm-bundles. Let $\supSWtot(\cB)$ be the supremum of \ETW achieved by \bunPolicies that can only use arm-bundles from $\cB$. There exists a distribution $D$ over $\cB$ such that
$T\cdot \SW(D)\geq \supSWtot(\cB)$ and $x_j(D)\leq s_j$ for each good $j$.
\end{lemma}

\noindent Here $T\cdot \SW(D)$ is seen as an approximation for $\SWtot(D)$, given that $x_j(D)\leq s_j$ for each $j$.

\xhdr{Reduction to best fixed bundle.}
A distribution $D$ over arm-bundles can be replaced by arm-bundle $\bar{x}= \Ex{x\sim D}{x}$, in the sense that $\SW(\bar{x})\geq \SW(D)$ (because is concave in the expected bundle, see \Cref{caveman}), and $x_j(D)=\bar{x}_j$ for each good $j$.
\footnote{A similar statement --- that any distribution over arms is ``dominated" by some arm --- is false for many other special cases of \BwK, including another version of dynamic pricing with limited supply \citep{BwK-focs13}.}

Let $\supSWtot$ be the supremum of \ETW over all \bunPolicies $\pi$. By \Cref{cl:limited-supply-A-to-pi}, it is also the supremum of \ETW over all pricing policies.

We claim that for any given $\eps>0$ there exists a finite set $\cB$ of arm-bundles such that
\begin{align*}
    \supSWtot(\cB)\geq \supSWtot - \eps.
\end{align*}
This holds because $\SW(x)$ is a H\"{o}lder-continuous function of bundle $x$ (see \Cref{holderman}), and so for any $\epsilon$, there is a fine enough discretization of bundles that yields an $\epsilon$-net for social welfare.

Putting this together, we reduce arbitrary \bunPolicies to fixed-bundle policies:

\begin{corollary}\label{cor:LP-value}
For each $\eps>0$ there exists is an arm-bundle $x$ such that
\begin{align*}
T\cdot\SW(x)\geq \supSWtot-\eps
    \quad \text{and}\quad
x_j\leq s_j \;\text{for each good $j$}.
\end{align*}
\end{corollary}
\noindent Again, $T\cdot \SW(x)$ is seen as an approximation for $\SWtot(x)$, given that $x_j\leq s_j$ for each $j$.

\xhdr{Completing the proof.} Fix parameter $\alpha>0$. Applying \Cref{cor:LP-value} with $\eps=\alpha T/2$ and letting $p$ be a price vector that induces bundle $x$, we obtain a price vector $p$ such that
\begin{align*}
T\cdot\SW(p)\geq \supSWtot-\alpha T/2
    \quad \text{and}\quad
x_j(p)\leq s_j \;\text{for each good $j$}.
\end{align*}

The algorithm from \Cref{thm:intro-limited-expected} can compute a price vector $p^*$ such that
\begin{align*}
\SW(p^*)\geq \SW(p) -\alpha/2
    \quad \text{and}\quad
x_j(p^*)\leq s_j \;\text{for each good $j$}.
\end{align*}
In particular,
    $T\cdot \SW(p^*)\geq \supSWtot-\alpha T$.
It remains to bound the difference between the \ETW $\SWtot(p^*)$ and the estimate $T\cdot \SW(p^*)$:

\begin{lemma}\label{lm:limited-deviations}
Let
    $\Smin = \min_j s_j$
and assume that $T\,\Smin >32\,\log T$. Then
$$\SWtot(p^*) \geq T\cdot \SW(p^*) - O(\sqrt{T\log (T)/\Smin}).$$
\end{lemma}

(The lemma applies to all price vectors $p^*$ such that $x_j(p^*)\leq s_j$ for each good $j$. Its proof is deferred to \Cref{pf:lm:limited-deviations}.) Putting this all together, we see that
$$\SWtot(p^*) \geq \supSWtot-\alpha T - O(\sqrt{T\log (T)/\Smin}).$$
This completes the proof of \Cref{THM:INTRO-LIMITED-SUPPLY}.

\subsection{Indivisible goods and unit demands: proof sketch of of \Cref{THM:INTRO-LIMITED-SUPPLY}(b)}

Recall that we compete against a weaker benchmark: the best fixed distribution $\cD$ over the price vectors; more precisely, against 
    $\supSWtot := \sup_{\cD:\; x(\cD)\leq s} \SWtot(\cD) $.

We can bound the deviation between $\SWtot(\cD)$ and $T\cdot \SW(\cD)$ via the following lemma (which is stated and proved similarly to \Cref{lm:limited-deviations}).

\begin{lemma}\label{lm:limited-deviations-D}
Let
    $\Smin = \min_j s_j$
and assume that $T\,\Smin >32\,\log T$. Let $\cD$ be a distribution over price vectors such that $x(\cD)\leq s$. Then:
\[ \left| \SWtot(\cD) - T\cdot \SW(\cD) \right| \leq  O(\sqrt{T\log (T)/\Smin}).\]
\end{lemma}

Fix $\eps>0$ and choose some distribution $\cD$ such that
    $x(\cD)\leq s$
and 
    $\SWtot(\cD)\geq \supSWtot -\eps$.
Apply \Cref{THM:INTRO-LIMITED-SUPPLY} to construct a distribution $\cD^*$ such that
    $x(\cD^*)\leq s$
and
    $\SW(\cD^*)\geq \SW(\cD)-\alpha$.
By \Cref{lm:limited-deviations-D}, it follows that
\begin{align*}
    \SWtot(\cD)- \SWtot(\cD^*)
        &\leq T(\SW(\cD)-\SW(\cD^*)) +  O(\sqrt{T\log (T)/\Smin}) \\
        &\leq \alpha T +  O(\sqrt{T\log (T)/\Smin}).
\end{align*}

\subsection{Proof of \Cref{lm:limited-deviations}}
\label{pf:lm:limited-deviations}

Let $\pi$ denote the the fixed-price policy with price vector $p^*$. For the sake of the argument, let us consider the execution of $\pi$ in the problem instance $\cI$ with time horizon $T$, but without the supply constraint. Let $Z_t$ be the realized total welfare of this execution by time $t$. Without loss of generality, we view an execution of $\pi$ in the original problem instance as an execution in the unlimited-supply instance, truncated at round $\tau$ when the original problem instance would halt. Thus, the total realized welfare of $\pi$ in the original problem instance is $Z_\tau$, where $\tau$ is a stopping time.

Let $x_j = x_j(p^*)$ be the expected consumption of a given good $j$. Let $y_{j,t}$ be the realized total consumption of this good by time $t$. Let $w = \SW(p^*)$ be the expected per-round welfare for $p^*$. By Chernoff Bound, letting
    $c_0 = \sqrt{8\, \log T}$,
with probability at least $1-T^{-2}$ we have
\begin{align}\label{eq:limited-chernoff}
|y_{j,t}-t x_j| \leq c_0\,\sqrt{t x_j}
    \quad\text{and}\quad
    Z_t \geq wt - c_0\sqrt{t}
    \quad \text{for each good $j$ and all rounds $t\leq T$}.
\end{align}

An execution of $\pi$ on unlimited-supply problem instance $\cI$ is called \emph{clean} if the event in \cref{eq:limited-chernoff} holds. To prove the lemma, it suffices to show in clean execution,
\begin{align}\label{eq:limited-deviations-Ztau}
Z_\tau \geq T\cdot \SW(p^*) - O(\sqrt{T\log (T)/\Smin}).
\end{align}
So we will assume a clean execution from now on.

Let $B_j=s_j\, T$ be the supply for good $j$. The stopping time $\tau$ can be expressed as
\begin{align}\label{eq:limited-deviations-tau}
\tau = \min_{\text{goods $j$}} \min(T,\tau_j),
\quad\text{where}\quad
\tau_j = \min\left\{\text{rounds $t$}: y_{j,t}>B_j\right\}.
\end{align}
Informally, we can think of each $\tau_j$ as the stopping time for good $j$. Let us analyze $\tau_j$.

\begin{claim}\label{cl:limited-deviations-tau}
$\tau_j \geq \frac{B-c_0\sqrt{B_j}}{x_j}$, for each good $j$.
\end{claim}
\begin{proof}
Let $\eps = c_0\sqrt{B_j}/x_j$. It suffices to prove that for each round $t\leq B_j/x_j-\eps$ we have $y_{j,t}\leq B_j$. This is so because by \cref{eq:limited-chernoff} we have
    $ y_{j,t} \leq t x_j + c_0\,\sqrt{t x_j} \leq B_j - \eps x_j + c_0\sqrt{B_j} \leq B_j$.
\end{proof}

\begin{claim}\label{cl:limited-deviations-mintau}
$\min(T,\tau_j)\geq T-2 c_0 \sqrt{T/s_j}$, for each good $j$.
\end{claim}
\begin{proof}
If $T \leq \frac{B-c_0\sqrt{B_j}}{x_j}$, then the claim follows trivially by \Cref{cl:limited-deviations-tau}. Else, we have
$$ T \geq \frac{B-c_0\sqrt{B_j}}{x_j} \geq \frac{B_j}{2x_j},
    \quad\text{so}\quad
\tau_j \geq T-\frac{c_0 B_j}{x_j} \geq T-2c_0 T/\sqrt{B_j} = T-2c_0 \sqrt{T/s_j}. \qedhere
$$
\end{proof}

Plugging this into \cref{eq:limited-deviations-tau}, it follows that
    $\tau\geq T-2c_0\sqrt{T/\Smin}$.
Since
    $Z_\tau \geq w\tau - c_0\sqrt{\tau}$
by \cref{eq:limited-chernoff}, the lower bound on $\tau$ implies \cref{eq:limited-deviations-Ztau}. This completes the proof of \Cref{lm:limited-deviations}.


\section{Conclusions and open questions}
\label{sec:conclusions}

We provide a polynomial time dynamic pricing algorithm for maximizing welfare over the allocation of $d$ goods, for buyers satisfying reasonable assumptions on their valuation functions. Prior work either required explicit assumptions on the aggregate price response function unsupported by micro-economic foundations, or had running time exponential in $d$.

Let us highlight two interesting directions. First: give an algorithm with a more reasonable polynomial run time. While we achieve polynomial dependence on $d$ and $\alpha$, this is mainly a proof of concept result: the degree of the polynomial run-time of our algorithm is quite high. A (much) smaller degree is desirable, but appears beyond the reach of our current techniques. Are there practical algorithms that achieve the same guarantees?  Second: extend our results to revenue optimization, a more traditional objective in the dynamic pricing literature. With our current techniques, this extension requires a major assumption on valuations, namely that buyers' valuations are \emph{uniformly homogeneous} with degree $m < 1$: that there exists a constant $m < 1 $ such that for every buyer $i$, bundle $x$, and scalar $\lambda$,
    $v_i(\lambda\, x)=\lambda^m v_i(x)$ (the extension to revenue maximization subject to this assumption follows from the techniques of \cite{RUW16}). Can revenue maximization be handled subject to weaker assumptions?


\section*{Acknowledgements.} We are grateful to Moshe Babaioff for
valuable feedback on an early draft of this paper. This work was
partially supported by NSF grant CNS-1253345, a Sloan Foundation
Fellowship, and a DARPA grant. Parts of this work have been done while
Zhiwei Steven Wu was visiting Microsoft Research.

\newpage
\bibliographystyle{acmtrans}

\bibliography{refs,bib-abbrv,bib-slivkins,bib-bandits,bib-AGT,bib-ML}

\appendix

\section{Basic Definitions and Tools}

The following condition is useful for obtaining the Slater's condition
and strong duality in convex program.

\begin{definition}[Non-empty interior]\label{def:strictdude} Let $C\in \RR^d$ be a convex
  set defined by the following set of constraints:
    \begin{align*}
      f_i(x) \leq 0, \qquad& i = 1, \ldots , k\\
      g_j(x) \leq 0, \qquad& j = 1, \ldots, m
    \end{align*}
    where each $f_i$ is a convex function and each $g_j$ is an affine
    function.  We say that the set $C$ has a~\emph{non-empty interior}
    if there exists a point $x^*\in C$ such that $f_i(x) < 0$ for each
    $i$.
\end{definition}

\begin{lemma}
\label{lem:sconvex}
Let $\phi\colon C \rightarrow \RR$ be a $\sigma$-strongly convex function, and let $x^* = \argmin_{x \in C} \phi(x)$ be the minimizer of $\phi$.  Then, for any $x\in C$,
\iffull \[ \fi
\|x - x^*\|^2_2 \leq \frac{2}{\sigma} \cdot (\phi(x) - \phi(x^*)).
\iffull \] \fi
\iffull
Similarly, if $\phi$ is $\sigma$-strongly concave, and $x^* = \argmax_{x \in C} \phi(x)$, then for any $x \in C$,
\iffull \[ \else $ \fi
\|x - x^*\|_2^2 \leq \frac{2}{\sigma} \cdot (\phi(x^*) - \phi(x)).
\iffull \] \else $ \fi\fi
\end{lemma}

\begin{theorem}[Real-valued Additive Chernoff-Hoeffding Bound]\label{chern}
  Let $X_1, \ldots , X_m$ be i.i.d. random variables with
  $\Ex{}{X_i} = \mu$ and $a\leq X_i\leq b$ for all $i$. Then for every
  $\alpha > 0$,
\[
  \Pr\left[ \left| \frac{\sum_i X_i}{m} - \mu\right| \geq
    \alpha\right] \leq 2 \exp\left( \frac{-2 \alpha^2 m}{(b -
      a)^2}\right).
\]
\end{theorem}

\section{Proof for Noisy Gradient Descent}
\label{app:noisyman}
\noisyman*

\begin{proof}
  Let $x^*\in \argmin_{x\in C}c(x)$. For each time step $t$, we will
  write $\hat x_{t+1} = \Pi_C (y_{t+1})$. Using the basic property of
  convexity and the elementary identity
  $2\langle a, b\rangle = \|a\|^2 + \|b\|^2 - \|a - b\|^2$, we can
  derive the following
\begin{align*}
  c(x_t) - c(x^*) &\leq \langle  g_t , x_t - x^*\rangle\\
                  &= \frac{1}{\eta} \left\langle (x_t - y_{t+1}), (x_t - x^*)\right\rangle\\
                  &= \frac{1}{2\eta} \left(\|x_t - y_{t+1}\|^2 + \|x_t - x^*\|^2 - \|y_{t+1} - x^*\|^2\right)\\
                  &= \frac{1}{2\eta} \left(\|x_t - x^*\|^2 - \|y_{t+1} -x^*\|^2 \right) + \frac{\eta}{2} \|g_t\|^2
\end{align*}

Since $x^*\in C$, by the property of the projection mapping $\Pi$, we
know that
\[
  \|y_{t+1} - x^* \|^2 \geq \|\Pi_C(y_{t+1}) - x^*\|^2 = \|x_{t+1} -
  x^* + \hat x_{t+1} - x_{t+1}\|^2 \geq \|x_{t+1} - x^*\|^2 - 2E \|x_{t+1} - x^*\|
\]
Since the set $C$ has diameter bounded by $D$, it follows that
\begin{align*}
  c(x_t) - c(x^*) &\leq\frac{1}{2\eta} \left( \|x_t - x^*\|^2 - \|x_{t+1} -x^*\|^2 + 2E\|x_{t+1} - x^*\|  \right) +
                    \frac{\eta}{2} \|g_t\|^2\\
                  &\leq \frac{1}{2\eta} \left( \|x_t - x^*\|^2 - \|x_{t+1} -x^*\|^2 \right) +
                    \frac{\eta G^2}{2}  + \frac{DE}{\eta}
\end{align*}

Also note that $\|g_t\| \leq G$, then summing over the
resulting inequality over all time steps $s$ yields
\[
  \sum_{t=1}^T (c(x_t) - c(x^*)) \leq \frac{D^2}{2\eta} + \frac{\eta
    G^2}{2} + \frac{DE}{\eta}
\]
Therefore, if we set the step size $\eta = D/(G\sqrt{T})$, we get the
following guarantee using convexity
\[
c(z) \leq \frac{1}{T}\sum_{t=1}^T c(x_t) \leq DG/\sqrt{T} + GE\sqrt{T},
\]
which recovers our claim.
\end{proof}

\section{Missing Details and Proofs of~\Cref{MAINSEC}}




 \ohfun*
\begin{proof}
  First, we will show that
  $g_{\hat x}(\hat p) = \min_{p\in \RR_+^d} g_{\hat x}(p) = \OPT(\hat
  x)$.  Note that
  \begin{align*}
    g_{\hat x}(\hat p)= \max_{x\in \cF^n} \cL_{\hat x}(x, \hat p) &= \max_{x\in \cF^n}\sum_{i} \psi(v_i) \, v_i(x_{i}) -
                                                           \sum_{j=1}^d \hat p_j \left( \sum_{v_i\in \cV} \psi(v_i) \, x_{ij} - \hat
                                                           x_j\right)\\
                                                         &= \max_{x\in \cF^n} \sum_{i} \psi(v_i) \left( v_i(x_{i}) - \langle \hat p_j,  x_{i}\rangle + \langle \hat p, \hat x\rangle\right)
  \end{align*}
  By~\cref{nips}, we have $x^\bullet_i =  \argmax_{x_i\in \cF} \left[v_i(x_i) - \langle \hat
      p, x_i\rangle + \langle \hat p, \hat x\rangle\right]$.
  This means $g_{\hat x}(\hat p) = \cL_{\hat x}(x\bl, \hat p)$.
  Since for each $j\in [d]$, we have
  $\sum_i \psi(v_i) x\bl_{ij} = \hat x_j$, we get
  \[
    g_{\hat x}(\hat p) = \cL_{\hat x}(x\bl, \hat p) = \sum_i \psi(v_i) \,
    v_i(x_{i}\bl) \leq \OPT(\hat x)
  \]
  where the last inequality follows from the fact that $x\bl$ is a
  feasible solution to the program $\SCP(\hat x)$. Therefore, we must
  have $g(\hat p) = \OPT(\hat x)$, and so $\hat p$ is an optimal dual
  solution to the convex program.  Since $\hat p$ is the optimal dual
  solution, it follows from~\cref{nips} that $x\bl$ is the optimal
  primal solution. Finally, the uniqueness of the primal solution
  follows from the fact that the objective in~\cref{ha} is strongly
  concave. 
\end{proof}

\begin{claim}\label{spiderman}
  Let $\hat x \in \left(\cF\cap \RR_{>0}^d\right)$ be an inducible
  bundle, induced by price vector $\hat p$. Then
  $\SW(\hat p) = \SW(\hat x)$.
\end{claim}

\begin{proof}
  For each $i$, let $x\bl_i(\hat x)$ be the allocation to buyer $i$ in
  the convex program $\SCP(\hat x)$. By~\Cref{ohfun}, we know that
  $x\bl_i(\hat x) = x^*_{v_i}(\hat p)$, and also
  $\OPT(\hat x) = \Ex{v_i\sim\psi}{v(x\bl_i(\hat x))}$. By linearity
  of expectation,
  \[
    \SW(\hat p) = \Ex{v_i\sim \psi}{v(x\bl_i(\hat x)) - \langle c,
      x\bl_i \rangle} = \OPT(\hat x) -  \langle c, \hat x\rangle,
  \]
  which recovers our claim.
\end{proof}

\caveman*
\begin{proof}
  Let $\hat x, \hat y\in \cF$, and let $\hat z = (\hat x + \hat
  y)/2$. It suffices to show that
\[
  \SW\left(\hat z \right) \geq \frac{\SW(\hat x) +
    \SW(\hat y)}{2}.
\]
Now for each $v_i\in \cV$, let $x^\bullet_i$, $y^\bullet_i$,
$z^\bullet_i$ be the bundles assigned to buyer $i$ in the optimal
solutions of the convex programs $\SCP(\hat x)$, $\SCP(\hat y)$ and
$\SCP(\hat z)$. Furthermore, let
$r^\bullet_i = (x^\bullet_i + y^\bullet_i)/2$.
Observe that the assignment $r\bl$ is also a feasible solution for the
convex program $\SCP(\hat z)$, so by the optimality of $z\bl$, we have
\[
\Ex{v_i\sim \psi}{v_i(z\bl_i)} \geq \Ex{v_i\sim \psi}{v_i(r\bl_i)}
\]
By the definition of $\SW$ and $\OPT$, we have
\[
\SW(\hat x) = \Ex{v_i\sim \psi}{v_i(x\bl_i)} - \langle c, \hat x\rangle, \quad \mbox{ and }\quad
\SW(\hat y) = \Ex{v_i\sim \psi}{v_i(y\bl_i)} - \langle c, \hat y \rangle.
\]
It follows that
\begin{align*}
  \SW(\hat x) + \SW(\hat y) &= \Ex{v_i\sim \psi}{v_i(x\bl_i)} - \langle c, \hat x\rangle  +\Ex{v_i\sim \psi}{v_i(y\bl_i)} - \langle c, \hat y \rangle\\
&= \Ex{v_i\sim \psi}{v_i(x\bl_i)}  +\Ex{v_i\sim \psi}{v_i(y\bl_i)} - 2 \langle c, \hat z \rangle\\
(\mbox{by concavity of } v_i) \quad &\leq 2 \Ex{v_i\sim \psi}{v_i(r\bl_i)} - 2 \langle c, \hat z \rangle = 2 \SW(\hat z)
\end{align*}
which completes our proof.
\end{proof}

We can translate the H\"{o}lder continuity condition on the
valuation functions in class $\cV$ to a H\"{o}lder continuity
condition of the function $\OPT$ and $\SW$, which will be useful for
our analysis later.

\begin{restatable}{lemma}{holderman}\label{holderman}
  $\OPT \colon  \cF \rightarrow \RR$ as defined
    in~\cref{bundleform} is $(d^{1-\beta} \lambda , \beta)$-H\"{o}lder
    continuous w.r.t. the $\ell_1$ norm.
\end{restatable}

\begin{proof}
  Let two bundles $\hat x, \hat y\in \cF$ be such that
  $\hat x_j = \hat y_j + \Delta$ and $\hat x_{j'} = \hat y_{j'}$.  We
  will first show that
  $\OPT(\hat x) - \OPT(\hat y) \leq \lambda \|\hat x - \hat
  y\|_2^\beta = \Delta^\beta$.

  For each $i$, let $x^\bullet_i$ and $y^\bullet_i$ be the assigned
  bundles to buyer $i$ in $\SCP(\hat x)$ and $\SCP(\hat y)$
  respectively. By our assumption, we know
  $\sum_i\psi(v_i) x\bl_i = \hat y_j + \Delta$. Then there exists
  $\delta_1, \ldots, \delta_n \in \RR_+$ such that
  $\sum_{i = 1}^n \psi(v_i)\delta_i = \Delta$ and
  $x\bl_{ij} \geq \delta_i$. Now consider a vector $x'\in \cF^{nd}$ such
  that for each $i$,
\[
  x'_{ij} = x\bl_{ij} - \delta_i \geq 0
  \qquad \mbox{ and } \qquad
  x'_{ij'} = x\bl_{ij'} \quad \mbox{for each }j'\neq j.
\]
Note that $x'$ is a feasible solution to the convex program
$\SCP(\hat y)$, so we have
$\sum_i \psi(v_i) v_i(x'_i) \leq \OPT(\hat y)$. Furthermore,
\begin{align*}
  \sum_i \psi(v_i) v_i(x'_i)              &=  \sum_i \psi(v_i) v_i(x'_i) - \left( \OPT(\hat x) - \OPT(\hat x) \right)\\
                                          &\geq \OPT(\hat x) - \left|\OPT(\hat x) - \sum_i \psi(v_i) v_i(x'_i) \right|\\
                                          &= \OPT(\hat x) - \left|\sum_i \psi(v_i) (v_i(x\bl_i) - v_i(x'_i)) \right|\\
  \mbox{(H\"{o}lder continuity of $v_i$)} \qquad &\geq\OPT(\hat x) - \lambda\left|\sum_i \psi(v_i) \, \delta_i^\beta \right|\\
  \mbox{(Jensen's inequality)} \qquad &\geq \OPT(\hat x) - \lambda \left(\sum_i \psi(v_i) \delta_i \right)^\beta
                                        = \OPT(\hat x) - \lambda \Delta^\beta\\
\end{align*}
Therefore, we have shown that the values of $\OPT$ on any two bundles
$\hat x$ and $\hat y$ that differ by one coordinate satisfy H\"{o}lder
continuity condition. To show this condition for any two bundles
$x, y\in \left(\cF \bigcap \RR_{>0}^d \right)$, we can apply the same
argument iteratively. First, we can form a set of $(d + 1)$ bundles
$b^0 = x, b^1, \ldots, b^d = y$ such that bundles $b^j$ and $b^{j+1}$
differ by at most one coordinate, and
$$\sum_{r=1}^{d} \|b^r - b^{r-1}\|_1 = \|x - y\|_1$$ Let
$\Delta_r = \|b^r - b^{r+1}\|_1$ for each $r$, we can then write
\begin{align*}
  |\OPT(x) - \OPT(y)| &\leq \sum_{r=0}^{d - 1} |\OPT(b^r) - \OPT(b^{r+1})|\\
                      &\leq \lambda \sum_{r=0}^{d - 1} \Delta_r^\beta \leq \lambda \, d \left(\frac{1}{d}\sum_{r=0}^{d - 1} \Delta_r\right)^\beta = d^{1- \beta}\lambda \|x - y\|_1^\beta\\
\end{align*}
where the last inequality follows from applying {Jensen's inequality}.
\end{proof}

\subsection{\Cref{sec:bunprice}}

\sixers*

\begin{proof}
  Since both $\cF^n$ and $\cP(\eps)$ are convex and $\cP(\eps)$ is compact,
  by Sion's minimax theorem~\cite{sion1958}, there is a
  value $\ROPT$ such that
\begin{equation}\label{lakers}
  \max_{x\in \cF^n} \min_{p \in \cP(\eps)} \cL_{\hat x}(x, p) = \min_{p
    \in \cP(\eps)}\max_{x\in \cF^n} \cL_{\hat x}(x, p) = \ROPT
\end{equation}
Since $\cP(\eps) \subseteq \RR_+^d$, we must have
$\ROPT \geq \OPT(\hat x)$. Now we only need to show that
$\ROPT \leq \OPT(\hat x) + \alpha$ with $\alpha = \eps^2\sigma
/4$. Let $(x^\star, p^\star)$ be a pair of minimax strategies of
\cref{lakers}. That is
\[
  x^\star \in \argmax_{x\in \cF^n} \min_{p\in \cP(\eps)} \cL(x, p)
  \qquad \mbox{ and } \qquad p^\star \in \argmin_{p\in \cP(\eps)}\max_{x\in
    \cF^n} \cL(x, p)
\]
It suffices to show that
$\cL(x^\star, p^\star) \leq \OPT(\hat x) + \alpha$. Suppose not, then
we have
\begin{align*}
  \OPT(\hat x) + \alpha   &<  \min_{p\in
                            \cP(\eps)}\cL(x^\star, p)\\
                          & = \sum_{i}\psi(v_i)\,v_i(x^\star_i) - \max_{p\in \cP(\eps)} \sum_{j\in
                            [d]} \left(\sum_{i} \psi(v_i)\, x^\star_{ij} - \hat x
                            \right)p_j\\
                          & \leq \sum_{i}\psi(v_i)\,v_i(x^\star_i)
\end{align*}

Now consider the bundle $y$ such that
$y_j = \max\{\sum_i \psi(v_i) x^\star_{ij}, \hat x_j\}$. For each $i$,
let $y\bl_i = x_{v_i}^*(p^*(y))$ and
$x\bl_i = x_{v_i}^*(p^*(\hat x))$. By definition,
$\OPT(y) = \sum_i \psi(v_i) v_i(y\bl_i)$.  It is clear that
$$
\OPT(y) \geq \sum_i \psi(v_i) v_i(x^\star_i) > \OPT(\hat x).
$$
Let $L = \lambda^{1/\beta} \left(\frac{4d}{\eps^2\sigma}
\right)^{(1-\beta)/\beta}$, then we can construct the following price
vector $p'\in \cP(\eps)$ such that $p_j' = L$ for each good $j$ with
$\sum_{i} \psi(v_i)\, x^\star_{ij} > \hat x_j$, and $p'_j = 0$ for all
other goods. By~\Cref{holderman}, we can derive the following
\begin{equation}\label{warriors}
\OPT(y) - \OPT(\hat x) \leq d^{1-\beta}\lambda \|y - \hat x\|_1
\end{equation}

It follows that
\begin{align}
  \OPT(\hat x) + \alpha < \cL(x^\star, p^\star) &\leq \cL(x^\star, p')\\
                                                &= \sum_i \psi(v_i) v_i(x^\star_i) - \left\langle p', \sum_i \psi(v_i)x^\star_i - \hat x \right\rangle\\
                                                &= \sum_i \psi(v_i) v_i(x^\star_i) - \sum_{j: y_j > \hat x_j} L\, (y_j - \hat x_j)\\
                                                &\leq  \OPT(y) - L \| y - \hat x\|_1 \label{cavs}
\end{align}

By combining \cref{warriors,cavs}, we obtain
\begin{align*}
\alpha < d^{1-\beta}\lambda \|y - \hat x\|_1^\beta - L\|y - \hat x\|_1 = \lambda \|y - \hat x\|_1^\beta \left(d^{1-\beta} - \frac{L}{\lambda} \|y -\hat x\|_1^{(1-\beta)} \right)
\end{align*}
Since $\alpha > 0$, we know that
$\|y - \hat x\|_1 < d\left(\frac{\lambda}{L} \right)^{1/(1-\beta)}$. By
our setting of $L$,
\[
  \alpha < d \lambda \left(\frac{\lambda}{L} \right)^{\beta / (1-\beta)}
  = \frac{\eps^2\sigma}{4} = \alpha
\]
which is a contradiction. Therefore, the minimax value $\ROPT$ is no
more than $\OPT(\hat x) + \alpha$.
\end{proof}

 \westbrook*

\begin{proof}
  By~\Cref{ohfun}, we know that
  $\OPT(\hat x) = \min_{p\in \RR_+^d} g_{\hat x}(p)$.
We will abuse notation to write $x'_i = x_{v_i}^*(p')$ for each buyer
of type $i$. By the result of~\Cref{sixers}, we also
have
\[
g_{\hat x}(p') = \cL(x', p') \leq \ROPT + \alpha \leq \OPT(\hat x) + 2\alpha.
\]
For each $i$, let $x_i\bl= x_{v_i}\bl(\hat x)$. Then note that
$\cL(x\bl, p') = \OPT(\hat x) - \langle p' , \hat x - \hat x\rangle =
\OPT(\hat x)$. Moreover, $x'$ is the maximizer for $\cL(\cdot, p')$,
so it follows that
\begin{align}\label{jordan}
0 \leq \cL(x', p') - \cL(x\bl, p') \leq 2\alpha.
\end{align}

For each buyer of type $i$, any bundle and price vector $p$, we will
write $u_i(x, p) = v_i(x) - \langle p, x \rangle$ to denote the
quasilinear utility function of buyer $i$. By our assumption on the
valuation functions in $\cV$, we know that $u_i(\cdot, p)$ is a
$\sigma$-strongly concave function over the space $W$ for any price
vector $p$. We know from \cref{jordan} that
\[
  \cL(x', p') - \cL(x\bl, p') = \sum_i \psi(v_i) [u_i(x'_i, p') -
  u_i(x\bl_i, p')] \leq 2\alpha.
\]
Since for each type $i$, the bundles $x'_i$ and $x\bl_i$ lie in the
set $W$, we have the following based on~\Cref{lem:sconvex}
\[
  \frac{\sigma}{2} \|x\bl_i - x'_i\|_2^2 \leq u_i(x'_i, p') -
  u_i(x\bl_i , p')
\]
Also, by Jensen's inequality, we have
\[
  \|\hat x - x'\|_2^2 = \left\| \sum_i \psi_i (x\bl_i - x'_i)
  \right\|_2^2 \leq \sum_i \psi_i \| x\bl_i - x'_i \|_2^2
\]
It follows that
\[
\frac{\sigma}{2} \|\hat x - x'\|_2^2 \leq 2\alpha
\]
and so $\|\hat x - x'\|_2\leq 2\sqrt{\alpha/\sigma}$.
\end{proof}

\unbiasedking*

\begin{proof}
  Given $x\bl = \argmax_{x \in \cF^n} \cL(x, p)$, we know by the
  Envelope theorem that the gradient of $g$ can be obtained as
\begin{align}
  \frac{\partial g_{\hat x}(p)}{\partial p_j} = \hat x_j - \sum_i \psi(v_i) \, x\bl_{ij}, \qquad \mbox{for each } j \in [d] \label{durant}.
\end{align}
By the definition of $x\bl$ we have,
\begin{align*}
  x\bl &= \argmax_{x\in \cF^n} \left[\sum_i \psi(v_i) v_i(x_i) - \sum_{j\in [d]} p_j \left(\sum_i \psi(v_i) x_{ij} - \hat x_j \right) \right]\\
&=  \argmax_{x\in \cF^n} \left[\sum_i \psi(v_i) v_i(x_i) - \sum_{j\in [d]} p_j \left(\sum_i \psi(v_i) x_{ij} \right) \right]\\
&=  \argmax_{x\in \cF^n} \left[\sum_i \psi(v_i) \left( v_i(x_i) - \langle p, x_{i} \rangle \right) \right]\\
\end{align*}
Observe that the decision variables $x_i$ for each buyer $i$ is
independent of all other buyers' decision in the $\argmax$ expression
above. This means
\[
x\bl_i = \argmax_{x_i \in \RR_+^d} \left[ v_i(x_i) - \langle p, x_i \rangle \right]
\]
It follows that for each buyer of type $i$, the bundle $x\bl_i$
corresponds to her purchased bundle under prices $p$, that is
$x_{v_i}^*(p)$. By combining with \cref{durant}, we now have that
\[
  \nabla g_{\hat x}(p) = \Ex{v\sim \psi}{\hat x - x_v^*(p)} = \hat x - x_\psi^*(p),
\]
which completes our proof.
\end{proof}

 \harden*

\begin{proof}
  The algorithm will start by running $\log(1/\delta)$ independent
  instantiations (indexed by $l$) of noisy gradient descent method
  (\Cref{hellyea}) to optimize the function $g_{\hat x}$. We will set
  number of iterations $T_1$ and learning rate $\eta$ to be
\[
  T_1 = \frac{4096 dL^2 R^2}{\eps^4 \sigma^2} \qquad \mbox{ and }\qquad
  \eta = \frac{R}{L \sqrt{dT}}
\]
with
$L = \lambda^{1/\beta} \left(\frac{4d}{\eps^2\sigma}
\right)^{(1-\beta)/\beta}$.  At each iteration point $p$ by the
gradient descent method, we will obtain an unbiased estimate of the
gradient by first posting the prices $p$, observe the bundle $x^*(p)$
purchased by a random buyer, and then compute the vector
$(\hat x - x^*(p))$ as the estimate.  Since both $\hat x$ and $x^*(p)$
satisfy $\|\hat x\|, \|x^*(p)\|\leq R$ and
$\hat x, x^*(p)\in \RR_+^d$, we have
\[
\|\hat x - x^*(p) \| \leq R
\]
Furthermore, by our choice of search space $\cP(\eps)$, each
iteration point $p$ satisfies $\|p\| \leq L\sqrt{d}$. By the guarantee
of~\Cref{hellyea}, each instantiation of the noisy gradient descent
method will output a price vector $p(l)$ such that
\[
\Ex{}{g_{\hat x}(p(l))} \leq \min_{p\in \cP'(\eps)} g_{\hat x}(p) + \frac{\eps^2 \sigma}{64}.
\]
By Markov inequality, for each instantiation $l$,
\[
\Pr\left[g_{\hat x}(p(l)) - \min_{p\in \cP'(\eps)} g_{\hat x}(p) \geq \frac{\eps^2 \sigma}{32}\right] \leq 1/2.
\]
Since we have $\log(1/\delta)$ independent instantiations of the noisy
gradient descent method, with probability at least $1 - \delta$,
there exists an instantiation $l^*$ that outputs a price vector
$\hat p$ such that
\[
  g_{\hat x}(\hat p) - \min_{p\in \cP'(\eps)} g_{\hat x}(p) \leq
  \frac{\eps^2 \sigma}{32}.
\]
By~\Cref{westbrook}, we have that
\[
\|x^*_\psi(\hat p) - \hat x \| \leq \eps.
\]
The number of rounds of interactions with the buyers is bounded
by $$O\left(\frac{dL^2 R^2\log(1/\delta)}{\eps^4 \sigma^2} \right),$$
which recovers our claim.
\end{proof}

\begin{algorithm}[h]
  \caption{Learning the price vector to induce a target bundle:
    $\learnprice(\hat x, \eps, \delta)$}
 \label{alg:learnprice}
  \begin{algorithmic}
    \STATE{\textbf{Input:} A target bundle $\hat x\in C$, target
      accuracy $\eps$, and confidence parameter $\delta$}

    \INDSTATE{Initialize: restricted price space
      $\cP(\eps)=\left\{p\in \RR_+^{d} \mid \|p\|\leq
        \sqrt{d}L \right\}$ where
      \[
        L = \left(\lambda\right)^{1/\beta}\,\left(
          \frac{4d}{\eps^2\sigma} \right)^{(1-\beta)/\beta}, \quad
        T_0= \log(2/\delta), \quad T_1= \frac{16384\, d\,L^2
          R^2}{\eps^4 \sigma^2}, \quad T_2 = \frac{4R^2 \log(16d
          R^2/\delta)}{\eps^2}, \quad \eta= \frac{R}{L\sqrt{dT_1}}\qquad
      \]
    }
    \INDSTATE{For $l = 1, \ldots , T_1$:}
    \INDSTATE[2]{Set $p^1_j = 1 \mbox{ for all good }j \in [d]$}
    \INDSTATE[2]{For $t = 1, \ldots , T_0$:}
    \INDSTATE[3]{Query $\rep$ with $p^t$ and observe an unbiased estimate for $x^*_\psi(p^t)$}
    \INDSTATE[3]{Update price vector with gradient descent:
      \[
        \hat{p}^{t+1}_j = p^t_j - \eta \left(\hat x_j -
          x^*(p^t)_j\right) \mbox{ for each }j\in [d],\qquad p^{t+1} =
        \Pi_{\cP(\eps)} \left[\hat p^{t+1}\right]
\]
}
    \INDSTATE[2]{Let $p(l) = 1/T\sum_{t=1}^T p^t.$}
    \INDSTATE[2]{For $s = 1, \ldots , T_2$:}
    \INDSTATE[3]{Query $\rep$ and observe an unbiased estimate $x(s)$for  $x_\psi^*(p^t)$}
    \INDSTATE[2]{Let $x(l) = \frac{1}{T_2} \sum_{s = 1}^{T_2} x(s)$}

    \INDSTATE{Let $l^* = \argmin_{l\in [T_1]} \| x(l) - \hat x \|$ }
    \STATE{\textbf{Output:} $\hat p = p(l^*)$}
    \end{algorithmic}
  \end{algorithm}

 \bunprice*

\begin{proof}
  We will first run the subroutine in~\Cref{harden} with accuracy
  parameter $\eps/4$ and confidence parameter $\delta/2$. This will
  output a list of $\log(2/\delta)$ prices $P$ such that with
  probability at least $1-\delta/2$, there exists $p^*\in P$ such that
  \begin{equation}
    \| \hat x - x^*_\psi(p^*) \| \leq \eps/4.\label{irving}
  \end{equation}
  For the remainder of the proof, we will condition on this accuracy
  guarantee, which is the case except with probability $\delta/2$.

  Next, for each of the price vector $p\in P$, we will post the price
  $p$ and observe purchased bundles by random buyers under prices $p$
  for $T_2$ number of rounds where
  $T_2 = \frac{16 d R^2 \log(4d/\delta)}{\eps^2}.$ For each good $j$,
  the amount of good $j$ purchased by a random buyer is a random
  variable with in the range of $[0, R]$ by our assumption on the
  class of valuations $\cV$. Let $\overline x_j$ be the average amount
  of $j$-th good purchased by these $T_2$ number of buyers. Note that
  $\Ex{\psi}{\overline x_j} = (x^*_\psi(p))_j$ for all good $j$.  By
  applying Chernoff bound (\Cref{chern}), we have with probability
  $1 - \delta/(2d \log(2/\delta))$ that
\[
  \| \overline x_j - (x^*_\psi(p))_j\| \leq \frac{\eps}{4\sqrt{d}}
\]
By applying a union bound over all price vectors in the set $P$ and
all goods $j\in [d]$, we have with probability at least $1 - \delta/2$ that
\begin{equation}
\|\overline x - x^*_\psi(p)\| \leq \frac{\eps}{4}\label{thom}
\end{equation}
We will condition on this event for the rest of the proof, which is
the case except with probability $\delta/2$. Let $x'$ be the average
bundle for posting the price vector $p^*$. Applying triangle
inequality on~\cref{irving,thom}, we get
\begin{equation}
\|x' - \hat x \| \leq \eps/2.\label{bogut}
\end{equation}
Finally, we will select the price vector $\hat p\in P$ with its
average bundle $\overline x$ closest to target bundle $\hat x$
w.r.t. $\ell_2$ distance. Based on~\cref{bogut}, we must have
\[
\|\overline x - \hat x \| \leq \eps/2.
\]
By applying triangle inequality on~\cref{thom}, the output price
vector must satisfy
\[
\| \hat x - x_\psi^*(\hat p) \| \leq \eps,
\]
which recovers our claim.
\end{proof}


%

\subsection{\Cref{awesomeness}}

\envelope*

\begin{proof}
First, recall that in \cref{bundleform}, $\SW$ is defined as
\[
  \SW(\hat x) = \OPT(\hat x) - \langle c , \hat x\rangle.
\]
It follows that to prove our claim, we just need to show that
$\hat p\in \partial \OPT(\hat x)$. Let
$\hat y\in \RR_{>0}^d \cap \cF$. By the definition of subgradient,
it suffices to show that
\[
  \OPT(\hat x) + \langle \hat p, (\hat y - \hat x) \rangle \geq
  \OPT(\hat y).
\]
Let $x\bl$ and $y\bl$ be the optimal solutions for $\SCP(\hat x)$ and
$\SCP(\hat y)$ respectively. This allows us to derive
\begin{align*}
  \OPT(\hat x) + \langle \hat p, (\hat y - \hat x) \rangle &= \sum_i \psi(v_i) v_i(x\bl_i) + \langle \hat p, (\hat y_j - \hat x_j)\\
                                                           &= \sum_i \psi(v_i) v_i(x\bl_i) + \sum_j \hat p_j \left(\sum_i \psi(v_i) (y\bl_{ij} - x\bl_{ij}) \right)\\
&= \sum_i \psi(v_i)\left( v_i(x\bl_i) - \langle \hat p, x\bl \rangle \right) + \sum_j \hat p_j \left( \sum_i \psi(v_i) y\bl_{ij}\right)
\end{align*}
Note that for each $i$, $x\bl_i = x^*_{v_i}(\hat p)$ by~\Cref{ohfun}, so we
must have for each $i$
$$
v_i(x\bl_i) - \langle \hat p, x\bl_i \rangle \geq v_i(y\bl_i) - \langle \hat p, y\bl_i \rangle.
$$
It follows that
\[
  \OPT(\hat x) + \langle \hat p, (\hat y - \hat x) \rangle \geq
 \sum_i \psi(v_i)\left( v_i(y\bl_i) - \sum_j \hat p_j y\bl_{ij}  \right) + \sum_j \hat p_j \left( \sum_i \psi(v_i) y\bl_{ij}\right) = \sum_i \psi(v_i) v_i(y\bl_i) = \OPT(\hat y),
\]
which recovers our claim.
\end{proof}

\captainamerica*

\begin{proof}
  Let $\pi$ be an lottery allocation. Let $\tilde x\in \cF^n$ be an
  allocation such that each buyer is allocated the average bundle
  $\tilde x_i = \Ex{v_i\sim\psi}{\pi(i)}$ in $\pi^\star$. Since each
  buyer's valuation is concave, by Jensen's inequality, we have
  \[
    \sum_i\psi(v_i)\, [v_i(\tilde x_i) - \langle c, \tilde x_i\rangle]
    \geq \Ex{v_i\sim \psi, \pi}{v(\pi(i)) - c(\pi(i))}.
  \]
  In other words, the optimal welfare $\lot$ can also be achieved by a
  deterministic allocation. Moreover, the optimal welfare by a
  deterministic allocation is $\max_{x\in S} \SW(x)$. Therefore, we
  only need to bound the difference between the welfare of the
  resulting price vector and $\max_{x\in S} \SW(x)$.

  To do that, we will need to bound three different types of error,
  and show that they add up to at most $\alpha$ with probability at
  least $1 - \beta$. Before we proceed into analyzing each type of
  error, we condition on the event that all $(T+1)$ instantiations of
  $\learnprice$ achieve their target accuracy guarantees. Since each
  of them has confidence parameter $\delta$, by union bound, we know
  this is the case except with probability $(T + 1)\delta' = \delta$.

  First, we will show that the bundle $\overline x$ computed by the
  noisy subgradient descent method satisfies
  \begin{equation}\label{firstbaby}
    \SW(\overline x) \geq \max_{x\in S_\xi} \SW(x) - \alpha/2.
  \end{equation}
  Observe that the perturbation error in our subgradient descent is
  precisely the inducing error of $\learnprice$, which is bounded by
  $\eps$. By~\Cref{noisyman} and our settings of $T, \eta$ and $\eps$
  (in \Cref{alg:optprice}), we recover the bound in~\cref{firstbaby}.

  Next, since are not optimizing over the entire set $S$, we need to
  bound the loss in welfare for only optimizing over $S_\xi$. The
  result of~\Cref{tmac} can be applied, and by our choice of $\xi$ we
  get
  \[
    \max_{x\in S} \, \SW(x) - \max_{x'\in S_\xi} \, \SW(x') \leq \alpha/4.
  \]
 It follows that
  \begin{equation}\label{secondbaby}
    \SW(\overline x) \geq \max_{x\in S_\xi} \,\SW(x) -\alpha/2 \geq  \max_{x'\in S}\, \SW(x') - 3\alpha/4.
  \end{equation}

  Finally, in our last instantiation of $\learnprice$, we will learn a
  price vector $\hat p$ induce a bundle $\hat x$ that is close to
  $\overline x$, that is $\|\hat x - \overline x\| \leq \eps$. Given
  our choice of $\eps$, we can bound the difference in welfare as in
  the analysis of~\Cref{tmac} to get
  \begin{equation}\label{thirdbaby}
    \SW(\hat p) = \SW(\hat x) \geq \SW(\hat x) - \alpha/4.
  \end{equation}
  By combining~\cref{firstbaby,secondbaby,thirdbaby}, we have obtained
  \[
    \SW(\hat p) \geq \max_{x\in S} \SW(x) - \alpha.
  \]
  Note that $\op$ only interacts with the buyers through
  $\learnprice$, and each instantiation requires no more than
  $\poly(d, 1/\alpha, \log(1/\delta))$ rounds of interactions with the
  buyers.
\end{proof}

\begin{algorithm}[h]
  \caption{Learning the price vector to optimize social welfare:
    $\op(\alpha,\delta, s)$}
 \label{alg:optprice}
  \begin{algorithmic}

    \STATE{\textbf{Input:} Target accuracy $\alpha$, confidence
      parameter $\delta$, and per-round expected demand bound $s$}

 \INDSTATE{Initialize:
      \[
        \eps= \left(\frac{\alpha^2}{16 R S^2} \right)^{\beta/(5\beta -
          4)},\qquad M = \sqrt{d} (\lambda)^{1/\beta} \left(4d/\sigma
        \right)^{(1-\beta)/\beta}
    \]
    \[
      T = \frac{16R^2M^2}{\alpha^2 \eps^{(2-2\beta)/\beta}},\qquad
 \eta = \frac{R}{ \sqrt{T} M (1/\eps)^{(2- 2\beta)/\beta}}, \qquad \delta' = \delta
      / (T + 1), \qquad
      \xi = 
      \left( \frac{\alpha}{4\lambda d^\beta + \sqrt{d} \|c\|} \right)^{1/\beta}
      \]
      Restricted bundle space $S_\xi$, initial bundle $x^1 \in S_\xi$
      and prices $p^1 = \learnprice(x^1, \eps, \delta')$}
    \INDSTATE{\textbf{For} $t = 1, \ldots , T$:} \INDSTATE[2]{Let
      $y^{(t+1)} = \Pi_{S_\xi}[ x^{t} + \eta (p^{t} - c)]$}

    \INDSTATE[2]{Let $p^{t+1} = \learnprice(x^{t+1}, \eps)$}
    \INDSTATE{Let $\overline x = 1/T \sum_{t=1}^T x^t$}
    \STATE{\textbf{Output:} the last price vector
      $\hat p = \learnprice(\overline x, \gamma, \delta')$}
    \end{algorithmic}
  \end{algorithm}

\section{Missing Details in~\Cref{SEC:UNIT-DEMAND} }

\begin{lemma}
  The entropy function $H(x) = \sum_{j=1}^N x_i\log(1/x_i)$ is
  $(\sqrt{N}, 1/2)$-H\"{o}lder continuous w.r.t. $\ell_1$-norm over the simplex.
\end{lemma}

\begin{proof}
  We will first show that the function $f(t) = t \log(1/t)$ is
  H\"{o}lder continuous. The first-order derivative of the function is
  $f'(t) = \log(1/t) - 1$.

  For any pair $a, b\in [0, 1]$, we can write
  \[
    |f(a) - f(b)| \leq \int_a^b |f'(t)| \, dt \leq \left( \int_a^b
      |f'(t)|^2 \, dt \right)^{1/2} \cdot \left( \int_a^b 1^2 \,
      dt \right)^{1/2}
  \]
  where the last inequality follows from H\"{o}lder inequality.
  Moreover,
  $$\int_a^b |f'(t)|^2 \, dt \leq \int_0^1 |f'(t)|^2 \, dt = 1
  \qquad \mbox{ and }\qquad \int_a^b 1^2 \, dt = |a - b|$$ It follows
  that
\[
  |f(a) - f(b)| \leq |a - b|^{1/2}
\]
Therefore, the function $f$ is $(1, 1/2)$-H\"{o}lder continuous.

Now we will use this to show the H\"{o}lder continuity of the entropy
function. Let $x, y$ be two probability vectors in $N$-dimensional
space, and $\delta_j = |x_i - y_i|$, then 
\begin{align*}
  |H(x) - H(y)| &\leq \sum_{j=1}^N |f(x_j) - f(y_j)| \\
                & \leq \sum_{j=1}^N  |x_j - y_j|^{1/2} \\
                &= {N} \left(1/N\sum_{j=1}^N |x_j - y_j|^{1/2}
                  \right)\leq N \left(1/N\sum_{j=1}^N |x_j - y_j|\right)^{1/2}\\
                & =   \sqrt{N}(\|x - y\|_1)^{1/2}
\end{align*}
where the last inequality follows from Jensen inequality. 
\end{proof}

Then the H\"{o}lder continuity condition immediately follows from
above.

\begin{corollary}\label{lem:holderdude}
  Let function $\tilde v\colon \dd \rightarrow \RR_+$ be
  defined as $\tilde v(x) = \langle v, x\rangle + \eta H(x)$ for some
  $\eta >0$. Suppose that each coordinate of $v$ is non-negative and
  upper bounded by $\vmax$. Then the function is
  $(\sqrt{d+1} + \vmax, 1/2)$-H\"{o}lder continuous
  w.r.t. $\ell_1$-norm.
\end{corollary}

\trans*

\begin{proof}
  To show that each coordinate of $p'$ is non-negative, note that
  $p'_{d+1} = 0$ and
  \[
    p'_j = (p_j - p_{d+1}) - \min\{0 , \min_{j'\in [d]} (p_{j'} -
    p_{d+1})\} \geq (p_j - p_{d+1}) - \min_{j'\in [d]} (p_{j'} -
    p_{d+1}) \geq 0.
  \]

  To show that the preference of the buyer remains the same under the
  two price vectors, we consider two cases: buyer will either choose
  the last item, or select one of the first $d$ items.

  Suppose that we are in the first case. Then we know that
  $v_{d+1} - p_{d+1} \geq \max_{j'\in [d]} (v_{j'} - p_{j'})$, and it
  follows that
  $v_{d+1} - p'_{d+1} = v_{d+1} - (p_{d+1} - p_{d+1}) \geq \max_{j'\in
    [d]} (v_{j'} - (p_{j'} - p_{d+1}))$. Note that for each
  $j' \in [d]$, we also have $p'_{j'} \geq (p_{j'} - p_{d+1})$, and
  this means 
\[
  v_{d+1} - p'_{d+1} \geq \max_{j'\in [d]} (v_{j'} - p'_{j'}).
\]
Therefore, the buyer will continue to choose item $(d+1)$.

Next consider the second case in which the buyer chooses an item
$j^*\in [d]$ under the price vector $p$. Note when we construct $p'$,
the prices for the first $d$ items are translated by the same amount,
so we have
\[
  j^* = \argmax_{j\in [d]} [v_j - p_{j}] =  \argmax_{j\in [d]} [v_j - p'_{j}]
\]
Now it suffices to show that the buyer will not select the last item
under $p'$.  To show that, we will further consider two sub-cases:
either $\min_{j'\in [d]} (p_{j'} - p_{d+1}) < 0$ or
$\min_{j'\in [d]} (p_{j'} - p_{d+1}) \geq 0$. Suppose that
$\min_{j'\in [d]} (p_{j'} - p_{d+1}) < 0$. Let
$j^\triangle = \argmin_{j'\in [d]} (p_{j'} - p_{d+1})$, then
\[
  v_{j^*} - p'_{j^*} = v_{j^*} - p_{j^*} + p_{j^\triangle} \geq
  v_{j^\triangle} - p_{j^\triangle} + p_{j^\triangle} =
  v_{j^\triangle} > 0 = v_{d+1} - p'_{d+1},
\]
this means the buyer will strictly prefer $j^*$.  Suppose that
$\min_{j'\in [d]} (p_{j'} - p_{d+1}) \geq 0$. Let
\[
  v_{j^*} - p'_{j^*} = v_{j^*} - p_{j^*} + p_{d+1} \geq v_{d+1} -
  p_{d+1} + p_{d+1} = v_{d+1} - p'_{d+1}
\]
which also implies that the buyer will still prefer $j^*$ to the empty
item.
\end{proof}

\crucial*

\begin{proof}
Note that the distribution of prices $\cD(p)$ satisfies
\[
  x^*_{\tilde \psi}(p) = \Ex{\psi}{\argmax_{x\in \Delta_{d+1}}
    [\langle v - p, x\rangle + \eta H(x)]} = \Ex{\psi, p'\sim
    \cD(p)}{\argmax_{x\in \Delta_{d+1}} \left[\langle v - p',
      x\rangle\right]},
\]
which recovers our lemma.
\end{proof}

\begin{algorithm}[h]
  \caption{Learning the price vector to optimize social welfare for 
    unit-demand buyers and indivisible goods: $\opu(\alpha,\delta, s)$}
 \label{alg:optpriceUD}
  \begin{algorithmic}

    \STATE{\textbf{Input:} Target accuracy $\alpha$, confidence
      parameter $\delta$, and per-round expected demand bound
      $s$}

    \INDSTATE{Initialize:
      $\eps= \alpha/2, \qquad \eta = \alpha/(2\log(d+1))$}
    \INDSTATE{Run $\op(\eps, \delta, s)$; For each price vector $p$
      queried by $\op$, obtain an unbiased estimate for $x^*_{\tilde\psi}(p)$
      by calling $\simulate(p, \eta)$} \INDSTATE{Let $\hat p$ be the
      price vector output by $\op$} \STATE{\textbf{Output:} the
      distribution over prices $\hat \cD = \cD(p)$}
    \end{algorithmic}
  \end{algorithm}

\adderror*

\begin{proof}
  Let $\oeta$ be the welfare achieved by the optimal fixed price
  vector $p^\eta$ in the $\eta$-regularized problem. Then by
  definition of $\oeta$
\[
  \Ex{v_i\sim \psi}{\langle v - c, x^\star_i\rangle + \eta H(x^\star_i)} \leq \oeta
\]
Since the entropy function is non-negative, it follows that
$\Ex{v_i\sim \psi}{\langle v - c, x^\star_i\rangle} = \lot\leq
\oeta.$

Next, we will show that there exists a price vector $p^\eta$ induces
the optimal fractional allocation in the regularized problem. Let
$x^\eta$ be optimal fractional allocation for the $\eta$-regularized
problem such that
$x_{\tilde\psi}^\eta = \Ex{v\sim
  \tilde\psi}{x^\eta_v}$. By~\Cref{ohfun}, we know that there exists a
price vector that induces the expected $x_{\tilde\psi}^\eta$ and the
associated randomized allocation. By the accuracy guarantee of
$\hat p$, we have
\[
  \Ex{\tilde v\sim \tilde \psi}{ \tilde v(x^*_{\tilde v}(\hat p)) -
    \langle c, x^*_{\tilde v}(\hat p) \rangle } \geq \oeta - \eps \geq \lot -\eps.
\]
Furthermore, by~\Cref{MWnote,trans}, we can rewrite 
\[
  \Ex{\tilde v\sim \tilde \psi}{ \langle v-c, x^*_{\tilde v}(\hat p)
    \rangle } = \Ex{v\sim \psi}{ \langle v-c, x^*_{\tilde v}(\hat p)
    \rangle + \eta H(x^*_{\tilde v}(\hat p))} = \Ex{v\sim \psi}{
    \langle v-c, x^*_{v}(\hat \cD) \rangle} + \eta \Ex{v\sim \psi}{
    H(x^*_{\tilde v}(\hat p))}
\]
where $x^*_v(\cD) = \Ex{p\sim \cD}{x^*_v(p)}$ is the randomized
allocation $\cD$ induces for each valuation $v$.
Following from the fact the entropy
term in regularized valuation is bounded by $\eta \log(d+1)$, we have
\[
  \SW(\hat \cD) = \Ex{v\sim \psi}{ \langle v-c, x^*_{v}(\hat \cD)
    \rangle} \geq \lot - \eps - \eta \log(d+1).
\]
which recovers our bound.
\end{proof}


\end{document}